\documentclass[letterpaper,11pt]{article}
\usepackage[T1]{fontenc}
\usepackage[margin=1in]{geometry}
\usepackage{amsmath,amsfonts,amsthm}
\usepackage{thm-restate}
\usepackage{pifont}
\usepackage{soul}
\usepackage{graphicx}
\usepackage{enumerate}
\usepackage{bbm,bm}
\usepackage{verbatim}
\usepackage{diagbox}
\usepackage[hypertexnames=false,bookmarksnumbered=true,final]{hyperref}
\usepackage[capitalize,nameinlink]{cleveref}
\usepackage[dvipsnames]{xcolor}
\usepackage[linesnumbered,ruled,vlined]{algorithm2e}
\usepackage{appendix}
\usepackage{makecell}
\usepackage{enumitem}
\usepackage{tablefootnote}
\usepackage{tikz}
\usetikzlibrary{datavisualization}
\usepackage{xifthen}
\usepackage{subcaption}
\usepackage{hhline}
\usepackage{authblk}

\def\colorschemesepia{sepia}
\def\colorschemedark{dark}
\def\colorschemelight{light}

\ifx\colorscheme\undefined
\let\colorscheme\colorschemelight
\fi

\ifx\colorscheme\colorschemelight
\colorlet{textColor}{black}
\colorlet{bgColor}{white}
\fi

\ifx\colorscheme\colorschemesepia
\definecolor{textColor}{HTML}{433423}
\definecolor{bgColor}{HTML}{fbf0da}
\fi

\ifx\colorscheme\colorschemedark
\definecolor{textColor}{HTML}{bdc1c6}
\definecolor{bgColor}{HTML}{202124}
\definecolor{textRed}{HTML}{ff968c}  %
\definecolor{textGreen}{HTML}{70cc70}  %
\definecolor{textBlue}{HTML}{8cbcff}  %
\definecolor{textCyan}{HTML}{70cccc}  %
\definecolor{textMagenta}{HTML}{d982d9}  %
\definecolor{textYellow}{HTML}{bfbf69}  %
\definecolor{textPurple}{HTML}{c58af9}
\colorlet{textHeavy}{white}
\else
\colorlet{textRed}{red!50!black}
\colorlet{textGreen}{green!35!black}
\colorlet{textBlue}{blue!40!black}
\definecolor{textPurple}{HTML}{995272}
\colorlet{textHeavy}{textColor}
\fi

\ifx\colorscheme\colorschemelight\else
\pagecolor{bgColor}
\color{textColor}
\fi

\colorlet{dimColor}{textColor!50!bgColor}

\hypersetup{
    colorlinks,
    citecolor=textGreen,
    linkcolor=textBlue,
    urlcolor=textPurple,
    pdfpagemode=UseNone,
    pdfstartview=FitW
}

\renewcommand{\v}{\tau}

\let\eps\varepsilon
\newcommand*{\defeq}{:=}

\newcommand*{\wLoG}{without loss of generality}

\newcommand*{\boolOne}{\mathbbold{1}}  %

\newcommand*{\vecZero}{\mathbf{0}}

\newcommand*{\xmark}{\ding{55}}

\newcommand*{\abs}[1]{\lvert #1 \rvert}

\DeclareMathOperator*{\E}{\mathbb{E}}
\DeclareMathOperator*{\Var}{Var}

\DeclareMathOperator{\Dem}{\mathit{OPT}}
\DeclareMathOperator{\Pre}{Pre}

\newcommand*{\APX}{{\sf APX}}
\newcommand*{\NP}{{\sf NP}}

\newcommand*{\p}{\bm{p}}
\newcommand*{\x}{\bm{x}}
\newcommand*{\y}{\bm{y}}
\newcommand*{\z}{\bm{z}}
\newcommand*{\pvec}{\bm{p}}

\newcommand*{\xvec}{\bm{x}}
\newcommand*{\yvec}{\bm{y}}

\DeclareMathAlphabet{\mathbbold}{U}{bbold}{m}{n}

\newcommand*{\D}{\mathcal{D}}
\newcommand*{\I}{\mathcal{I}}
\renewcommand{\P}{\mathcal{P}}
\newcommand*{\Dcal}{\mathcal{D}}
\newcommand*{\Ical}{\mathcal{I}}
\newcommand*{\Ncal}{\mathcal{N}}

\newcommand*{\phat}{\widehat{p}}

\newcommand*{\pvechat}{\widehat{\pvec}}

\theoremstyle{plain}
\newtheorem{theorem}{Theorem}
\newtheorem*{theorem*}{Theorem}

\newtheorem{lemma}{Lemma}

\theoremstyle{definition}
\newtheorem{definition}{Definition}
\newtheorem{example}{Example}

\crefname{appsec}{Appendix}{Appendices}
\crefname{lemma}{Lemma}{Lemmas}
\crefname{theorem}{Theorem}{Theorems}
\crefname{definition}{Definition}{Definitions}
\crefname{fact}{Fact}{Facts}
\crefname{claim}{Claim}{Claims}
\crefname{proposition}{Proposition}{Propositions}
\tikzset{
  my grid style/.style={
    tikz/data visualization/all axes={
      grid,grid={
        minor={style={draw={textColor!10!bgColor}}},
        major={style={draw={textColor!10!bgColor}}},
      }
    }
  }
}

\allowdisplaybreaks
\makeatletter
\g@addto@macro{\UrlBreaks}{%
\do\/%
\do\a\do\b\do\c\do\d\do\e\do\f\do\g\do\h\do\i\do\j\do\k\do\l\do\m%
\do\n\do\o\do\p\do\q\do\r\do\s\do\t\do\u\do\v\do\w\do\x\do\y\do\z%
\do\A\do\B\do\C\do\D\do\E\do\F\do\G\do\H\do\I\do\J\do\K\do\L\do\M%
\do\N\do\O\do\P\do\Q\do\R\do\S\do\T\do\U\do\V\do\W\do\X\do\Y\do\Z%
\do\0\do\1\do\2\do\3\do\4\do\5\do\6\do\7\do\8\do\9%
}
\makeatother

\makeatletter
\@ifpackageloaded{enumitem}{%
\newenvironment*{tightemize}{\begin{itemize}[noitemsep]}{\end{itemize}}%
\newenvironment*{tightenum}{\begin{enumerate}[noitemsep]}{\end{enumerate}}%
}{%
\newenvironment*{tightenum}{\begin{enumerate}}{\end{enumerate}}%
}
\makeatother

\let\citep\cite
\let\citet\cite

\title{Revenue-Optimal Pricing for Budget-Constrained Buyers\\ in Data Markets%
\thanks{J.~Garg and E.~Sharma were supported by NSF grant CCF-2334461.
B.~R.~Chaudhury and J.~Song were supported by the NSF Career Award CCF-2441580.}}

\author[1]{Bhaskar Ray Chaudhury}
\author[1]{Jugal Garg}
\author[1]{Eklavya Sharma}
\author[1]{Jiaxin Song}

\affil[1]{University of Illinois, Urbana-Champaign}
\affil[ ]{\{braycha, jugal, eklavya2, jiaxins8\}@illinois.edu}
\date{\empty}

\begin{document}

\maketitle

\begin{abstract}
We study revenue-optimal pricing in data markets with rational, budget-constrained buyers.
Such a market offers multiple datasets for sale, and buyers aim to
improve the accuracy of their prediction tasks by acquiring data bundles.
The market's objective is to price datasets to maximize total revenue, considering that
buyers with quasi-linear utilities choose their bundles optimally under budget constraints.
We allow the buyers to purchase fractions of datasets,
and the amount they pay is proportional to the fraction they receive.
Although \emph{competitive equilibrium} gives revenue-optimal pricing in rivalrous markets
with quasi-linear buyers, we show that revenue maximization in data markets is APX-hard.
Despite the hardness, we design a 2-approximation algorithm when datasets arrive online,
and a $(1-1/e)^{-1}$-approximation algorithm for the offline setting.

\end{abstract}

\section{Introduction}
\label{sec:intro}

Data is the fuel that drives AI-ML technologies. The rapid integration of AI-ML technologies into nearly every industry, coupled with declining storage costs, has elevated \emph{data} to be one of the most valuable assets of the 21st century.
According to~\citet{AcumenDM}, the U.S. big data market alone is projected to reach approximately $\$ 473$ billion by 2030---a testament to data's growing strategic and economic importance.
As data becomes an increasingly important economic asset, developing a rigorous theory of data pricing is essential to establish the foundations of emerging data-driven economic paradigms.

In this paper, we study revenue-optimal pricing strategies for a centralized data marketplace that sells datasets to a group of budget-constrained buyers. The marketplace offers $m$ datasets, denoted by $\Dcal_1, \Dcal_2, \dots, \Dcal_m$, and serves $n$ buyers with respective budgets $b_1, b_2, \dots, b_n$. Each dataset $\Dcal_j$ consists of $s_j$ data records. Let $p_j \in \mathbb{R}_{\geq 0}$ be the price of each dataset $\Dcal_j$. Purchasing $z_j$ data records from $\Dcal_j$ requires a payment of $p_j\cdot(z_j/s_j)$.
This pricing model reflects common practices in some real data marketplaces.
For example, commercial platforms such as Snowflake Marketplace allow data providers to charge buyers on a usage-based basis, including per-query and per-row pricing for access to paid datasets~\cite{suger_snowflake_pricing}.
Similarly, third-party data marketplaces employ volume-based pricing schemes, with providers such as Bright Data charging per thousand records accessed~\cite{brightdata2026datamarketplaces}.

Once the marketplace sets a price $p_j \geq 0$ for each dataset $\Dcal_j$, each buyer demands a utility-maximizing combination of data records, possibly drawn from multiple datasets, subject to their budget constraint. The marketplace's goal is to determine optimal prices to maximize the total earned revenue. In order to state our model and results clearly, it is important to first understand the buyer's utility model over the available datasets and their underlying records.

\paragraph{Utility for data.}
In this paper, following Shannon's information-theoretic view, we model data as an asset that reduces uncertainty~\cite{veldkamp2023valuing,cover2005entropy}. We consider environments where buyers are machine learning (ML) agents seeking to improve the accuracy of their predictions, for which additional data enhances precision.

Each buyer $i \in N$ aims to estimate an unknown parameter $\theta_i$, representing a latent quantity of interest, e.g., future traffic flow, energy demand, or delivery time. Each data record serves as a digitized signal, providing noisy information about $\theta_i$. Signals from each dataset follow a fixed prior distribution determined by the corresponding data-generation process, and its relevance for estimating $\theta_i$, and signal distributions differ across datasets due to heterogeneity in data generation and relevance. For instance, a buyer training a model to forecast traffic congestion ($\theta_i$) may purchase vehicle location data from a navigation platform and aggregated mobility statistics from a ride-sharing service. Both provide noisy signals about the same underlying state but differ in informativeness.

A \textit{data bundle} $\xvec_i = (x_{i,1}, x_{i,2}, \dots, x_{i,m})$ represents the collection of data records acquired by buyer $i$ from multiple datasets, where $x_{i,j}$ denotes the quantity purchased from dataset~$j$. Upon observing the full set of signals $S(\xvec_i)$, the buyer updates their belief about the latent parameter~$\theta_i$ to the posterior distribution $\theta_i \mid S(\xvec_i)$. Following the literature on \emph{value of data}~\citep{baley2025data}, we define the buyer’s \emph{accuracy gain} from bundle~$\xvec_i$ as the reduction in uncertainty about~$\theta_i$:
\[ a_i(\xvec_i) = \E\!\left[\Pre(\theta_i \mid S(\xvec_i))\right] - \Pre(\theta_i), \]
where the precision of a random variable, $\Pre(\theta)$,
is defined as the inverse of its variance, i.e., $\Pre(\theta) = 1 / \Var(\theta)$.
Consistent with the standard literature on revenue maximization and auction theory~\citep{myerson1981optimal,klemperer2004auctions}, we assume that the buyer’s net utility from a data bundle~$\xvec_i$ equals their value for the accuracy gain minus the total payment. Let $\pvec \defeq (p_1, \ldots, p_m)$ be the vector of prices. Then, the buyer’s net utility from bundle~$\xvec_i$ is given by
\[ u_i(\xvec_i, \pvec) = \alpha_i \, a_i(\xvec_i) - \sum_j p_j \cdot (x_{i,j}/s_j), \]
where $\alpha_i$ represents the buyer’s valuation parameter for accuracy improvement.
We also note that more general signaling structures---such as those arising from complementary or correlated datasets---can give rise to richer utility functions that capture interactions between datasets. In fact, these utilities have been well studied in both economics and computation. %
For the purposes of this paper, we focus on the simpler utility form described above. This form captures the essential trade-off between the informativeness of data and its cost, while keeping the analysis tractable, providing a natural starting point for our results.

\paragraph{Optimal demand bundles and the revenue maximization problem.}
Once the prices are fixed, each buyer selects an optimal bundle of data records that maximizes her utility subject to her budget constraint. Formally, given the price vector $\pvec$, buyer~$i$'s optimum demand bundle $\Dem_i(\pvec)$ is defined as,
\begin{align*}
    \Dem_i(\pvec)
    ~=~
    \arg\max_{\yvec} \bigg\{ u_i(\yvec, \pvec)
    ~\big|~
    y_j \leq s_j \,\forall j, \,\sum_{j=1}^m p_j \cdot (y_j/s_j) \le b_i \bigg\},
\end{align*}
where $\yvec = (y_1, \dots, y_m)$ denotes the quantities of data records purchased from the available datasets, and the feasible set includes all \emph{affordable bundles}, i.e., bundles whose total cost does not exceed the buyer's budget.
The data marketplace, anticipating that buyers demand their optimum bundles, aims to determine the prices that maximizes its total revenue.
We highlight an important subtlety. Observe that each buyer’s optimal bundle is independent of the demands of other buyers. In other words, the amount of data records from a particular dataset that one buyer demands does not depend on how much of that dataset is demanded by others. \emph{However, at any price vector $\pvec$, it is feasible to allocate every buyer their individually optimal demand bundle whenever it is well-defined.}
This is a direct consequence of the \emph{non-rivalrous} nature of data: consumption by one buyer does not diminish its availability to others.
Crucially, this property fails in traditional settings with \emph{rivalrous} goods, where it may be impossible to allocate every buyer their optimal bundle simultaneously at certain prices. Hence, for rivalrous goods, the set of feasible prices consists precisely of those at which it is possible to allocate all buyers their optimal bundles simultaneously.
Although this may seem to make the rivalrous case more intricate, we show in \cref{sec:overview} that, on the contrary, the opposite holds.

The resulting bi-level optimization problem can be expressed as

\begin{equation}
\label{pgm:rev-max}
\max_{\substack{\pvec \in \mathbb{R}_{\ge 0}^m \\ \xvec_i \in \Dem_i(\pvec) \,\forall i}}
\quad
\sum_{i=1}^n \sum_{j=1}^m p_j \cdot (x_{i,j}/s_j)
\end{equation}

Given the optimal prices $\pvec$ and the corresponding optimal demand bundles of buyers $\xvec = (\xvec_1, \xvec_2, \dots, \xvec_n)$, the pair $(\pvec, \xvec)$ constitutes a \emph{Stackelberg Equilibrium} (SE): the data market, acting as the leader, commits to a pricing strategy, and buyers, acting as followers, best respond with their utility-maximizing bundles.

\section{Overview of Our Results}
\label{sec:overview}

In this section, we outline our main contributions, where we provide high-level overviews and proof sketches;
full technical details are deferred to \cref{sec:approx-algo,sec:inapprox,sec:market-clear}.

For all our results, following the model in~\cite{ChaudhuryGMS26}, we assume that each latent parameter $\theta_i$ is drawn from a prior Gaussian distribution $\Ncal(0, \tau^{-1}_i)$. Each data record from dataset $\Dcal_j$ provides a noisy signal to buyer $i$ of the form $s_{i,j} = \theta_i + \eta_{i,j}, \quad \eta_{i,j} \sim \Ncal(0, \tau_{i,j}^{-1})$, where $\eta_{i,j}$ represents the observational noise. The noise variance $\tau_{i,j}^{-1}$ differs across datasets, reflecting the varying relevance and informativeness of each dataset for buyer $i$’s prediction task. In Appendix~\ref{app:form_of_accuracy}, we show that under the foregoing assumptions, we have $a_i(\xvec_i) = \sum_j \tau_{i,j} x_{i,j}$, i.e., the precision improves linearly with the number of data records from each dataset. We define the value of a data bundle $\bm x_i$ to buyer $i$ as $\alpha_i \cdot \sum_j \tau_{i,j} x_{i,j}$.

Assuming independent noise across datasets is realistic when each dataset provides a distinct proxy for the same latent quantity of interest. For example, if $\theta_i$ represents the true next-week demand for a product in a region, then transaction logs, search query volumes, and pre-order counts each provide imperfect measurements of $\theta_i$. The noise in each dataset is independent of the true demand and of the other datasets because it arises from separate operational mechanisms.

While one could consider correlated noises, or other generalizations leading to nonlinear data valuations---positive correlations inducing diminishing returns, negative correlations generating complementarities, or valuations being non-linear in precision---\emph{our focus is to show that even under linear valuations, which form a fundamental and extensively studied benchmark in economic theory~\cite{Gale60}, data markets differ substantially from classical markets in the complexity of the problem (\cref{sec:overview:differences}).}

For notational convenience, from here onwards, instead of representing the quantity of data in terms of the number of records,
we specify quantity in terms of \emph{fractions} of datasets, i.e, a bundle $\xvec_i = (x_{i,1}, \ldots, x_{i,m})$ contains an $x_{i,j} \in [0, 1]$ fraction of each dataset $j$.
If the price vector is $\pvec$, then the price of a bundle $\xvec_i$ would be $\pvec^T\xvec_i$.
Datasets typically have a very large number of records, so we allow each $x_{i,j}$ to be an arbitrary number in $[0, 1]$ (i.e., we \emph{do not} require $x_{i,j}$ to be a multiple of $1/s_j$).

For convenience, we abuse notation and denote the value of the entire dataset $\Dcal_j$ to buyer $i$ by $\tau_{i,j}$ instead of $\alpha_i\cdot\tau_{i,j}$,
and denote the value of a bundle $\xvec_i$ to buyer $i$ by $\tau_i(\xvec_i) = \sum_{j=1}^m \tau_{i,j}x_{i,j}$
instead of $\alpha_i\cdot\sum_{j=1}^m \tau_{i,j}x_{i,j}$.

\subsection{Difference from Rivalrous Markets}
\label{sec:overview:differences}

While the problem of maximizing revenue has been extensively studied in traditional rivalrous (goods) economies, the techniques and solution concepts developed there do not extend to our setting. We begin by highlighting the key technical barriers that prevent such extensions, and then introduce our proposed solutions.

\paragraph{Pricing rivalrous goods.}
The problem of pricing \emph{rivalrous goods} has been well-studied in the existing literature. Recall that the prices need to be set in such a way that it is feasible to give every buyer their optimum demand bundle (which is defined independently of the demand of other buyers). A closely related challenge also arises in the \emph{perfectly competitive} setting involving rivalrous goods. Here too, prices must be defined so that each buyer’s individually optimal bundle---determined independently of others’ demands---can be feasibly allocated. However, unlike in revenue maximization, the objective is not to choose prices that maximize the total revenue, but to identify prices that ensure \emph{market clearing}, i.e., that total aggregate demand equals total available supply for every good. Formally, the pair $(\p, \x)$ constitute a \emph{Competitive Equilibrium (CE)} if (1) every buyer receives their optimal bundle (i.e., $\x_i \in \overline{OPT}_i(\p)$)\footnote{In $\overline{OPT}_i(\p)$, buyers are supply-\emph{unaware}, i.e., $\overline{OPT}_i(\p) = \arg\max_{\yvec} \{ u_i(\yvec)\ |\ \sum_{j=1}^m p_j(y_j) \le b_i\}$}; (2) the market clears: the aggregate demand of every good equals its available supply, i.e., $\sum_{i} x_{i,j} = s_j$ for all $j$.
A CE is known to exist when buyers have quasi-concave utilities~\citep{arrow1954existence} and can be efficiently computed when buyers have linear utilities~\citep{eisenberg1959consensus, DevanurPSV08,Orlin10}.
Quite remarkably, the very same competitive pricing mechanism turns out to be \emph{revenue-optimal} when buyers have quasi-linear utilities~\citep{finster2023substitutes}---a rare and elegant coincidence where market efficiency and revenue optimality align perfectly. This striking harmony implies that computing the revenue-optimal prices in the rivalrous setting is, in fact, polynomial-time solvable since the convex formulation for CE with quasi-linear utilities~\citep{DGJMVY17} can be solved in polynomial-time.

\paragraph{Adapting CE does not work for non-rivalrous markets.}
A natural first step is to adapt the notion of a competitive equilibrium (CE) to the non-rivalrous setting, in the hope of obtaining similar guarantees---namely, that a CE implies revenue optimality, and can be computed in polynomial time. However, the standard definition of CE from rivalrous markets does not extend directly. In particular, the aggregate demand for a good can no longer be expressed as $\sum_i x_{i,j}$, since a non-rival good can be simultaneously consumed by multiple buyers.
In the non-rivalrous market by \cite{ChaudhuryGMS26}, the aggregate demand is defined as $\max_i x_{i,j}$ rather than $\sum_i x_{i,j}$, so the market clears if, for every dataset, at least one buyer includes all available records of that dataset in her optimal bundle.
Unfortunately, we find that this adaptation leads to undesirable outcomes. In particular, the total revenue generated by a CE can be substantially lower than that of a Stackelberg equilibrium (SE). Indeed, we construct an instance (\cref{ex:inapprox_n_ce_se}) that admits a unique CE whose revenue is far from optimal—demonstrating that, in the non-rivalrous setting, a CE may yield revenue that is not even a reasonable approximation of the maximum attainable revenue.

\begin{example}
\label{ex:inapprox_n_ce_se}
Consider a data market with $n$ buyers $a_1,\dots, a_n$ and one dataset $j_1$.
The first buyer has a budget of $2$ while each of the remaining $n-1$ buyers has a budget of $1.9$.
Each data buyer $a_i$ has a value of $\v_{i, 1}$ equal to the budget $b_i$.
At a CE, the price of $\D_{j_1}$ should be set at least $2$.
Otherwise, buyer $a_1$ will demand more than $1$, which exceeds the supply of dataset $j_1$.
However, none of the remaining $n-1$ buyers is willing to buy $\D_{j_1}$ anymore in that case, and the total revenue is at most $2$.
In contrast, the revenue-optimal price is $p_{j_1} = 1.9$ in an SE, and the revenue is $1.9\cdot (n-1)$.
Therefore, the gap between the two equilibria is $1.9\cdot(n-1)/ 2 \in \Omega(n)$.
\begin{figure}[htbp]
\centering
\begin{tikzpicture}[
    agent/.style={circle, draw, fill=blue!50!black!50, text=bgColor, minimum size=0.2cm, font=\small},
    job/.style={rectangle, draw, fill=green!50!black!50, text=bgColor, minimum size=0.2cm, font=\small},
    edge/.style={thick},
]

\node[agent] (a1) at (0, 2) {$a_1$};
\node[agent] (a2) at (0, 1) {$a_2$};
\node (a3) at (0, 0) {$\dots$};
\node[agent] (a4) at (0, -1) {$a_{n}$};

\node[job] (j1) at (2, 1.5) {$j_1$};
\draw[edge] (a1) -- (j1) node[midway, above] {\scriptsize 2};
\draw[edge] (a2) -- (j1) node[midway, above] {\scriptsize 1.9};
\draw[edge] (a4) -- (j1) node[midway, above] {\scriptsize 1.9};

\node at (2.75, 1.5) {\small \textcolor{red!75!textHeavy}{$\$2$}};
\node at (9.75, 1.5) {\small \textcolor{red!75!textHeavy}{$\$1.9$}};
\node at (-1, 2) {\small 2};
\node at (-1, 1) {\small $1.9$};
\node at (-1, 0) {\small $1.9$};
\node at (-1, -1) {\small $1.9$};

\node[agent] (a1r) at (7, 2) {$a_1$};
\node[agent] (a2r) at (7, 1) {$a_2$};
\node (a3r) at (7, 0) {$\dots$};
\node[agent] (a4r) at (7, -1) {$a_{n}$};
\node[job] (j1r) at (9, 1.5) {$j_1$};

\draw[edge] (a1r) -- (j1r) node[midway, above] {\scriptsize 2};
\draw[edge] (a2r) -- (j1r) node[midway, above] {\scriptsize 1.9};
\draw[edge] (a4r) -- (j1r) node[midway, above] {\scriptsize 1.9};

\node at (6, 2) {\small 2};
\node at (6, 1) {\small $1.9$};
\node at (6, 0) {\small $1.9$};
\node at (6, -1) {\small $1.9$};

\node[align=center] at (-2.5, 0.5) {\bf CE};
\node[align=center] at (4.5, 0.5) {\bf SE};
\end{tikzpicture}

\caption{Prices for CE and SE of \cref{ex:inapprox_n_ce_se}, where the prices are in \textcolor{red!75!textHeavy}{red}}
\label{fig:placeholder}
\end{figure}
\end{example}

\subsection{Buyer Behavior and Structured Solutions}
\label{sec:overview:rev-struct}

While the non-rivalry makes revenue maximization computationally harder than the rivalrous setting, it still offers a very clean closed-form characterization of the optimal revenue as function of the prices (unlike the rivalrous setting). To derive it, we must first understand how buyers behave.

After the price vector $\p = (p_1, \dots, p_m)$ is fixed,
each buyer faces a fractional knapsack problem, where the knapsack's capacity is $b_i$,
and each dataset $\Dcal_j$ has \emph{profit} $\v_{i,j} - p_j$ and price $p_j$.
Define dataset $j$'s bang-per-buck (for buyer $i$) as $\v_{i,j}/p_j$.
The optimal solution to the fractional knapsack problem is obtained via a greedy algorithm:
sort the datasets in non-increasing order of bang-per-buck,
and keep purchasing in that order till either the budget is exhausted,
or all datasets of bang-per-buck at least 1 have been purchased.
Thus, the revenue earned from each buyer $i$ is
\begin{equation}
\label{eq:ri}
r_i(\p) = \min\left(b_i, \sum_{j \in [m]:\,\v_{i, j}\ge p_j} p_j\right)
    = \min\left(b_i, \sum_{j=1}^m p_j\cdot\boolOne(\v_{i,j} \ge p_j)\right).
\end{equation}
(Here $\boolOne(X)$ is 1 if proposition $X$ is true and 0 otherwise.)
Denote the total revenue by $r(\p)$, and define it to be the sum of revenues from each buyer, i.e.,
\[ r(\p) \defeq \sum_{i=1}^n r_i(\p). \]
Note that non-rivalry ensures that the decision of one buyer does not constrain what other buyers can purchase, allowing this additive, per-buyer formula to capture the total revenue exactly. In contrast, rivalrous goods introduce interdependencies between buyers’ demands, which generally precludes such a simple closed-form expression.

Next, we show that we can restrict our attention to solutions
where each dataset's price equals some buyer's value for it.

\begin{lemma}
\label{thm:rdisc}
In a data market instance, for any price vector $\pvec \in \mathbb{R}_{\ge 0}^m$,
there exists another price vector $\pvechat \in \mathbb{R}_{\ge 0}^m$ such that
\begin{tightenum}
\item $r_i(\pvechat) \ge r_i(\pvec)$ for every buyer $i \in [n]$.
\item $\phat_j \in P_j \defeq \{\v_{i,j}: i \in [n]\}$ for every dataset $j \in [m]$.
\end{tightenum}
\end{lemma}
\begin{proof}
For any dataset $j \in [m]$, if $p_j > \max_{i=1}^n \v_{i,j}$, then no buyer will buy it,
so change its price to any $\phat_j \in P_j$.
Otherwise, increase its price from $p_j$ to the nearest element in $P_j$,
i.e., increase the price to $\phat_j \defeq \min(\{\v_{i,j}: i \in [n] \text{ and } \v_{i,j} \ge p_j\})$.

Then we have $\phat_j\cdot\boolOne(\v_{i,j} \ge \phat_j) \ge p_j\cdot\boolOne(\v_{i,j} \ge p_j)$
for all $i \in [n]$ and $j \in [m]$. Thus, by \cref{eq:ri}, we get
$r_i(\pvechat) \ge r_i(\pvec)$ for each buyer $i \in [n]$.
Additionally, we also have $\phat_j \in P_j$ for all $j \in [m]$.
\end{proof}

\Cref{thm:rdisc} gives us a way to model revenue-maximization as a discrete optimization problem.
In fact, it immediately gives us an $O(n^{m+1}m)$-time algorithm
for finding the revenue-maximizing price vector:
simply try all price vectors in $\prod_{j=1}^m \{\v_{i,j}: i \in [n]\}$
and output the one with the maximum total revenue.

\subsection{Computational Inapproximability}
\label{sec:overview:inapprox}

The $O(n^{m+1}m)$-time algorithm runs in polynomial time when $m$---%
the number of datasets---is constant.
However, if the number of datasets is large, the problem becomes \APX-hard.

\begin{restatable}{theorem}{thmLinearAPX}
\label{thm:apx_hardness}
Revenue maximization in data markets is \APX-hard.
\end{restatable}

This result underscores a striking contrast between rivalrous and non-rivalrous markets:
while the rivalrous settings allows for polynomial-time revenue-optimal solutions,
the non-rivalrous setting is hard even to approximate.

\subsection{Approximation Algorithm}
\label{sec:overview:approx-algo}

The closed-form expression for total revenue (\cref{eq:ri})
naturally leads to a greedy algorithm: initialize $\p = \vecZero$,
and iteratively raise the price of each dataset as much as possible---up to the point where
further increases would cause the revenue to decline.
If the revenue function is submodular (\emph{continuous submodular}), i.e., $r(\p +\Delta) - r(\p) \le r(\pvechat+\Delta) - r(\pvechat)$ for any $\p \succeq \pvechat$\,\footnote{An $r$-dimensional vector $a \succeq b$ if $a_i \geq b_i$ for all $i \in [r]$.} and $\Delta \ge \bm{0}$, then the greedy algorithm can achieve an approximation ratio of $2$ using a standard submodular analysis.
Unfortunately, we observe in \cref{example:non_sub} that the revenue function is not submodular: If most datasets are priced excessively high, buyers only spend their budgets toward the remaining lower-priced datasets, since $\v_{i,j} < p_j$ for all highly priced $j$. Consequently, the marginal revenue gain from increasing the price of a lower-priced dataset is greater when the other datasets are priced higher, contradicting submodularity.

\begin{example}[Non-submodularity of $r(\p)$]
\label{example:non_sub}
Consider a data market instance with two buyers and two datasets, and the parameters set as follows:
$\v_{1,1} = \v_{1,2} = 1, \v_{2, 1} = \eps, \v_{2,2}= 2$ and $b_1 = b_2 = 1$, where $0< \eps \ll 1$.
Then we have $r(\cdot)$ given by
\begin{alignat*}{2}
r(\eps, 1) &= \min(1, \eps+1) + \min(1, \eps+1) = 2 \quad & r(1, 1) &= \min(1, 2) + \min(1, 1) = 2 \\
r(\eps, 2) &= \min(1, \eps) + \min(1, 2) = \eps + 1 \quad & r(1, 2) &= \min(1, 1) + \min(1, 2) = 2\,.
\end{alignat*}
Notice that $r(1,2) - r(\eps, 2) = 1-\eps > r(1, 1) - r(\eps, 1)$ violates submodularity.
\end{example}

\paragraph{Reformulation using $k$-submodular functions.}
\Cref{example:non_sub} shows that one cannot, in general, expect diminishing marginal gains in revenue when varying the entire price vector.
However, submodularity might still hold in a restricted sense---specifically, when increasing the price of a single dataset while keeping all others fixed, as buyers will have decreasing residual budget with increasing prices.
This suggests the possibility of weaker notions of submodularity, such as $k$-submodularity.
To explore this, we first discretize the revenue function:
because of \cref{thm:rdisc}, we can assume \wLoG{} that
each dataset's price coincides with one of the buyers' value for it.
Hence, each price vector can alternatively be represented as a partition $(S_1, \ldots, S_n)$
of the $m$ datasets, where $j \in S_i$ implies that dataset $j$ is priced $\v_{i,j}$.
Assume $S_1, \dots, S_n$ are disjoint; if $p_j = \v_{i,j} = \v_{i',j}$ for two buyers $i$ and $i'$,
then assign dataset $j$ to either $S_i$ or $S_{i'}$ arbitrarily.
This way, we can replace the space of prices from $\mathbb{R}_{\ge 0}^m$
to the tuples $(S_1, \dots, S_n)$ of $n$-disjoint sets of $M$.

\emph{Surprisingly, after this transformation, we find some submodular properties of the induced new revenue function:}  Let $r(\mathcal{S})$ denote the revenue of the prices corresponding to $\mathcal{S}=(S_1, \dots, S_n)$, where we set $p_j = 0$ if $j\notin \bigcup_i S_i$.
The new revenue function is monotone since the revenue from a dataset is always zero if its price is zero.
Meanwhile, we notice that $r(\mathcal{S})$ also satisfies a weaker submodular property, $k$-submodularity, proposed in~\cite{huber2012towards}, which is equivalent to coordinate-wise diminishing returns when the function is also monotone.
Formally, using our notations,
\[ r(S_1, \dots, S_i\cup \{j\}, \dots, S_n) - r(S_1,\dots, S_n) \le r(T_1, \dots, T_i\cup \{j\}, \dots, T_n) - r(T_1,\dots, T_n) \]
when $j\notin \bigcup_{\ell}S_{\ell}$ and $T_{\ell} \subseteq S_{\ell}$ for all $\ell$.
Denote by $\p^S$ and $\p^T$ the price vectors corresponding to $(S_1, \dots, S_n)$ and $(T_1, \dots, T_n)$ respectively.
Unlike the counter-intuitive issue that happened in \cref{example:non_sub}---higher prices can lead to more remaining budgets, the condition $T_i \subseteq S_i$ actually offer a stronger guarantee of the usage of buyers' budget.
Since the price of every dataset in $\p^T$ is either same as in $\p^S$ or set to zero, each buyer in $\p^T$ has (weakly) more available budget than in $\p^S$, resulting in a weakly larger marginal revenue increase when $p_j$ is changed to $\v_{i, j}$.
Therefore, the new revenue function $r(\mathcal{S})$ is both monotone and $n$-submodular!
Hence, we can adapt the greedy algorithm for $k$-submodular maximization.
At a high level, the greedy algorithm first specifies an arbitrary order
of the datasets $\sigma = (j_1, \dots, j_m)$ and sets all the prices to zero.
Then in the $\ell$-th round, the algorithm sets the price of $p_{j_\ell}$ to one of
$\v_{1, j_\ell}, \dots, \v_{n, j_\ell}$ that maximizes the marginal increase of the total revenue.
The greedy algorithm is able to get an approximation ratio of 2 due to~\cite{ward2016maximizing}.
Since it can work with an arbitrary order of datasets, this algorithm also works when datasets arrive \emph{online}.
Using the algorithm of \citet{10.5555/2884435.2884465},
we can get a $(2-\Theta(1/n))$-approximate randomized algorithm.

\begin{theorem}
There exists a greedy 2-approximation algorithm to maximize revenue in data markets.
This algorithm also works when the datasets arrive online.
\end{theorem}

As shown in~\cref{lem:2_approx_linear}, the greedy algorithm can only achieve
an approximation ratio of at most $2 - \Theta(1/n)$ in the worst case for our revenue function.
Moreover, \cite{10.5555/2884435.2884465} shows an asymptotically tight
inapproximability ratio of $2-1/n$ for maximizing a general monotone $n$-submodular function.
Therefore, one cannot hope to get better than $2-\Theta(1/n)$ approximation for our revenue function by only using its $k$-submodularity. Fortunately, we observe that the inapproximability construction in~\cite{10.5555/2884435.2884465} relies on gadgets such as quadratic terms involving the number of elements not included in $S_1,\dots,S_n$, which are not suited for our setting.

\paragraph{Applying continuous greedy via function extension.}
Given this, one might wonder whether the above approximation ratio can be improved in our setting.
One possible idea is randomization.
We can create a variable $y_{j, i}$ to represent the probability that $j\in S_i$ or equivalently $p_j = \v_{i, j}$.
Let $D_{\y}$ be the distribution induced by the probabilities $(y_{j, i})_{j\in [m], i\in [n]}$.
Therefore, the objective function turns to the expectation of the total revenue, and the original optimization problem can be reformulated as follows:
\begin{equation}
\begin{aligned}
&\max_{\y} \Big(\mathop{\mathbb{E}}\limits_{\p\sim D_{\y}}\left[r(\p)\right]\Big) = \max_{\y} \Big(\mathop{\mathbb{E}}\limits_{\p\sim D_{\y}}\Big[\sum_{i=1}^n\min(b_i, \sum_{j} p_j\cdot \boolOne[p_j \le \v_{i, j}]) \Big]\Big)\\
&\text{subject to}\quad \y_j \in \Delta_n  \text{ for any } j \in [m]
\end{aligned}
\label{prob:op_exp}
\end{equation}
Note that the expectation operator is taken outside the minimum operator.
Since the minimum induces convexity in the revenue function, one cannot freely exchange their orders, i.e.,
\[
\mathop{\mathbb{E}}\limits_{\p\sim D_{\y}}\Big[\sum_{i=1}^n\min(b_i, \sum_{j} p_j\cdot \boolOne[p_j \le \v_{i, j}]) \Big]
\neq \sum_{i=1}^n\min\left(b_i, \sum_{j}  \mathop{\mathbb{E}}\limits_{\p\sim D_{\y}}\left[p_j\cdot \boolOne[p_j \le \v_{i, j}]\right]\right)\,.
\]
Therefore, Problem~\eqref{prob:op_exp} cannot be formulated as a linear program, which also matches our \APX-hardness result.  %

Another approach is to consider a distribution over deterministic prices.
The continuous greedy algorithm is particularly relevant, achieving a $(1 - 1/e)$ approximation for monotone submodular maximization, potentially improving the above ratio of 2.
For the problem $\max_{S\in \mathcal{I}} f(S)$, where $f:2^U\to \mathbb{R}_{\ge 0}$ is submodular and $\mathcal{I}$ is a matroid constraint, the algorithm defines the multilinear extension of $f$ as the expectation of $f$ over a distribution $F(\y) = \E_{S\sim D_{\y}}[f(S)] = \sum_{S\subseteq U} y_S\cdot f(S)$, where $\y = (y_S)_{S\subseteq U}$ denotes a probability distribution.
The algorithm starts with $\y(0) = \bm{0}$ and then updates $\y(t+\delta) = \y(t) + \delta \cdot \z$ by choosing the $\z \in \mathcal{I}$ that maximizes $\z\cdot \nabla F(\y(t))$.
It finally outputs $\y(1)$ as a randomized solution.
Unfortunately, continuous greedy does not work for $k$-submodular functions in general.
The main issue is that a $k$-submodular function $f$ only defines the values when $S_1,\dots, S_k$ are disjoint and the function values are undefined otherwise.
For example, it is unclear what the revenue should be defined as if a dataset is put into both $S_i$ and $S_{i'}$ (where $i \neq i'$). As a result, the gradient of $F$ can be undefined at some point. In \cref{sec:k_sub_sub_partition}, we provide an example illustrating that a $k$-submodular maximization problem cannot, in general, be modeled as a submodular maximization problem subject to a partition matroid constraint.

Interestingly, although such an extension does not always exist in the general monotone $k$-submodular functions, our $n$-submodular revenue function can indeed be extended to the entire domain, yielding $(1-1/e)^{-1} \approx 1.582$-approximation, as stated below:
\begin{restatable}{theorem}{thmContinuousGreedy}
\label{thm:continuous_greedy}
There exists a randomized greedy $(1-1/e)^{-1}$-approximate algorithm to maximize the revenue of a data market.
\end{restatable}

Given the set of datasets $M$, we create $n$ copies for every dataset in $M$.
Let $C^j = \{j^{(1)}, \dots, j^{(n)}\}$ be the set of $n$ copies for dataset $j$.
Let $\bar{M} = \bigcup_j C^j$.
Let $S$ be a subset of $\bar{M}$.
If $j^{(i)} \in S$, we interpret it as that price of $\D_j$ is set as $\v_{i, j}$.
Define a partition matroid $\mathcal{I}$ as follows
\[
\mathcal{I} = \{S: \abs{S\cap C^j }\le 1\text{ for any }  j\in M\},
\]
which enforces that each price $p_j$ can only be chosen from one of $\v_{1, j}, \dots, \v_{n, j}$.
Hence, every set $S$ within the matroid corresponds to valid prices.
Meanwhile, the optimal prices also corresponds to a set within a matroid.
\emph{As discussed above, the main challenge in applying the continuous greedy algorithm, lies in defining a submodular function  $\hat{r}$ which is also well-defined
{and poly-time computable}, for sets in $2^{\bar{M}} \setminus \mathcal{I}$}.
We next give a simple example of how we extend $r$.
Consider two buyers and one dataset.
Suppose $\v_{1, 1} < \v_{2, 1}$ without loss of generality.
Then, according to the definition of $n$-submodular revenue function, we have
\begin{align*}
r(\{1\}, \emptyset) = \min(b_1, \v_{1, 1}) + \min(b_2, \v_{1,1}),\quad   r(\emptyset, \{1\}) = \min(b_1, 0) + \min(b_2, \v_{2,1})
\end{align*}
Correspondingly, the values of $\hat{r}$ on the partition matroid $\mathcal{I} = \{\emptyset, \{1^{(1)}\}, \{1^{(2)}\}\}$ are given by
\begin{align*}
\hat{r}(\{1^{(1)}\}) = \min(b_1, \v_{1, 1}) + \min(b_2, \v_{1,1}),\quad   \hat{r}(\{1^{(2)}\}) = \min(b_1, 0) + \min(b_2, \v_{2,1})
\end{align*}
To define the value of $\hat{r}$ at $\{1^{(1)}, 1^{(2)}\}$, one possibility is to
merge the values of $\hat{r}(\{1^{(1)}\})$ and $\hat{r}(\{1^{(2)}\})$ by
adding the inner terms inside the budget-cap functions, e.g.,
$\hat{r}(\{1^{(1)}, 1^{(2)}\}) = \min(b_1, \v_{1, 1} + 0) + \min(b_2, \v_{1,1} + \v_{2,1})$.
Since the min function $\min(\cdot, \cdot)$ is submodular over $\mathbb{R} \times \mathbb{R}$,
it can be verified that the above extension is submodular over $\bar{M}$.
Following the above insights, we can define the following extension of the $n$-submodular function $r$ to the entire domain $2^{\bar{M}}$ as follows:
\begin{align}
\label{eq:extension_sub}
\hat{r}(S) = \sum_{i=1}^n \min\left(b_i, \sum_{j=1}^m \sum_{\ell=1}^n \v_{i, \ell, j} \cdot \boolOne[j^{(\ell)}\in S]\right)
\text{where }
\v_{i, \ell, j} = \begin{cases}
\v_{\ell, j}, & \text{if } \v_{\ell, j}\le \v_{i, j}\\
0, & \text{otherwise}
\end{cases}
\end{align}
For the same reason, function \eqref{eq:extension_sub} is submodular and also monotone over $2^{\bar{M}}$.
Therefore, by applying the continuous greedy algorithm~\cite{calinescu2011maximizing} for the following optimization problem
\begin{align}
\max_{S\subseteq 2^{\bar{M}}} \hat{r}(S) \quad \text{subject to } S\in \mathcal{I},
\end{align}
we are able to get a $(1-1/e)^{-1}$ approximation! It is also worth noting that, although extending a partial submodular function to the entire domain is hard in general~\citep{bhaskar2018partial}, our construction is based on a closed-form function and hence runs in polynomial time.

\subsection{Compatibility with \emph{Supply-Aware} CE (SACE)}
\label{sec:overview:market-clear}

Recall that in rivalrous markets with quasi-linear utilities, every competitive equilibrium (CE) is also revenue-optimal. However, as shown in \cref{ex:inapprox_n_ce_se}, there exist instances in the non-rivalrous setting where no CE is revenue-optimal, illustrating a fundamental incompatibility between CE and revenue optimality (or SE) in such markets.

This raises a natural question: can we define a relaxation of CE that remains compatible with an SE?
We answer this in the affirmative by introducing the notion of a supply-aware competitive equilibrium (SACE). In a SACE, each buyer selects an optimal bundle subject to both her budget constraint and the supply constraints on the number of data records available per dataset. Every CE is a supply aware CE, though the converse need not hold. The concept of supply-aware equilibria has been explored in different contexts. For instance, \cite{ConitzerKPSSSW22} and \cite{ChenKK24} consider it in the context of auctions, while \cite{MehtaS13} investigates it in \emph{market games} where buyers strategically report their utility functions to improve their outcomes. In \cref{sec:market-clear}, we establish that in data markets with quasi-linear utilities, there always exists a revenue-optimal solution (i.e., an SE) that is also a supply-aware CE. Moreover, we present a polynomial-time procedure that transforms any approximate SE into a supply-aware CE without any loss in revenue. This implies that a $(1 - 1/e)^{-1}$-approximate SE that is simultaneously a supply-aware CE can be computed in polynomial time. Detailed proofs and algorithmic guarantees are provided in \cref{sec:market-clear}.

\section{Further Related Work}
\label{sec:related-work}

\paragraph{Data Economics.}
With the proliferation of AI and data-centric technologies across industries, data economics has rapidly evolved into a vital area of study. A full survey is well beyond the scope of this paper, but we list the work that most closely resonates with our contributions. In this paper, we have adopted a very generic model for data valuations, defining an agent's utility as an improvement in prediction accuracy. There have been more domain-specific studies on how agents value data~\citep{farboodi2025valuing, veldkamp2023valuing, farboodi2023data}.  We refer the reader to~\citet{fleckenstein2023review} for a detailed review of data valuation methods.

Given the significance of data to ML tasks, it is no surprise that there has been a significant body of work on data acquisition mechanisms that incentivize agents to share data~\citep{fallah2024optimal, cummings2023optimal, fallah2022bridging, MurhekarYCLM23, acemoglu2023good, ghosh2011selling, nissim2014redrawing, chen2018optimal, varian2009economic}, by compensating them for their privacy loss. There have been studies on mechanisms that incentivize sellers to truthfully report the variances of their datasets to a data aggregator, who is tasked with attaining a certain variance level for her prediction~\citep{cummings2015accuracy}.

Data-markets are two-sided marketplaces that receive prediction requests from buyers and address these requests from the datasets hosted on them by data sellers. Several questions on truthfulness, incentives, and revenue maximization have been studied in data markets. There is a line of work~\citep{admati1986monopolistic, admati1990direct,  bergemann2018design, BabaioffKP12} that investigates revenue-maximizing strategies of a monopolist data seller.
\citet{AgarwalDS19} design a truthful mechanism for a centralized online data marketplace. Instead of pricing datasets, their model charges buyers proportional to the increase in their prediction task's accuracy. This makes the pricing scheme discriminatory (buyers who benefit more from the data pay more), but can sometimes make it easier for buyers to assess if a dataset is worth paying for.
There is also a line of work that discusses equilibrium and auctions in data markets in the presence of externalities~\citep{AgarwalDHR20, Hossain024,ichihashi2021competing}.
Several studies investigate the pricing of data/information from other first principles in different settings~\citep{mehta2021sell, pei2020survey, cai2020sell, bergemann2022economics,horner2016selling}.
In contrast to our monetary setting, prior work has investigated stable solutions in data exchange economies, where agents exchange data without monetary transfers \citep{BhaskaraGIKMS24,akrami2025theoretical,song2025existence}.

\paragraph{Submodular maximization.}
Our problem of revenue maximization is closely relevant to submodular maximization subject to a matroid constraint.
It is widely known that a simple greedy algorithm achieves an approximation ratio of $2$ \citep{nemhauser1978analysis}.
The greedy algorithm also achieves approximation ratio of $(1-\eps)$ for some particular matroids, e.g., when the matroid is uniform $\mathcal{I} = \{S: \abs{S} \le k\}$.
Later, \citet{calinescu2011maximizing} improved the approximation ratio for monotone submodular maximization from $2$ to $(1 - 1/e)$ by leveraging the pipage rounding technique of \citet{ageev2004pipage} and the continuous greedy process of \citet{vondrak2008optimal}.
Subsequently, \citet{feldman2011unified} introduced a unified continuous greedy framework that also applies to non-monotone submodular functions.
Their framework achieves an approximation ratio of roughly $1/e$ for the non-monotone case while keeping the $(1 - 1/e)$ guarantee for monotone submodular functions.
An important application is submodular maximization subject to a \emph{partition matroid} constraint~\citep{chekuri2004maximum,fleischer2006tight,vondrak2008optimal,azar2011submodular,mirrokni2008tight,garg2023approximating}, which also generalizes our revenue maximization problem: a uniform set $U$ is partitioned into $k$ sets $U_1,\dots, U_k$ associated with numbers $\ell_1, \dots, \ell_k$ and a set $X$ is independent if $\abs{X\cap U_i} \le \ell_i$ for any $i\in [k]$.
The \textsf{Submodular Welfare} problem~\citep{feige2006approximation,vondrak2008optimal,mirrokni2008tight} is a notable special case of submodular maximization under a partition matroid constraint and closely related to our problem: there are a set of of $m$ items and $n$ agents, each of whom is associated with a submodular and monotone utility function $f_i:2^{[m]} \rightarrow \mathbb{R}_{\ge 0}$.
The goal is to find a partition $(X_1, \dots, X_n)$ to maximize the social welfare $\sum_{i=1}^n f_i(X_i)$.
The problem can be approximated within a ratio of $(1-1/e)$ using the continuous greedy \citep{mirrokni2008tight} in the value oracle model and the ratio is optimal~\citep{khot2005inapproximability}.
Our problem is similar in that we also try to find a partition of the datasets and maximizing the sum of revenue.
However, unlike the \textsf{Submodular Welfare} setting, our revenue function $r$ cannot be simply decomposed into functions $r_1,\dots, r_n$ where $r_i$ solely depends on $X_i$, since setting some item $j$'s price to $v_{i, j}$ can also affect the revenue from buyer $i'$.
The problem of $k$-submodular function maximization, introduced in~\citet{huber2012towards} is closely relevant to the problem of submodular maximization subject to a partition matroid.
Intuitively, $k$-submodularity only requires submodularity within the partition matroid.
Our example in \cref{sec:k_sub_sub_partition} demonstrated that a $k$-submodular function cannot always be extended to a submodular function in general.
As a result, the approximation ratios for $k$-submodular maximization are typically worse.
\citet{10.5555/2884435.2884465} prove that one may need an exponential number of queries to get an approximation ratio better than $2$ even for monotone $k$-submodular functions.
In contrast, as mentioned above, the optimal ratio for submodular maximization subject to partition matroid is $(1-1/e)$.
For non-monotone $k$-submodular functions, \citet{ward2016maximizing} prove that the greedy algorithm can achieves $1/(r+1)$-approximation if the function is $r$-wise monotone.

On the negative side, when the matroid is defined by a cardinality constraint, \citet{10.1145/285055.285059} proves that the problem of finding the maximum $k$-cover is hard to approximate within a factor of $(1-1/e + \eps)$, which also matches the approximation ratio of $(1-1/e)$.
Moreover, most of the above approximation algorithms or hardness results are typically defined with access to a value-based query oracle for $f$ and a membership-based oracle for $\mathcal{I}$, where a problem is considered hard to approximate if it cannot be approximated using a polynomial number of queries unless $\textsf{P} = \textsf{NP}$.
Another way to model is to assume the function $f$ and $\mathcal{I}$ are provided as polynomial-size circuits.
\citet{dobzinski2012query} prove that the inapproximability for the value oracle model still holds even if the input function $f$ is provided as a polynomial-size circuit.

\paragraph{Competitive equilibrium.}
The problem of computing CE in traditional rivalrous economies has been extensively studied. We focus here on the most relevant work for \emph{Fisher markets} with linear utilities. In such a market, there is a set of buyers and a set of goods, where each buyer has a budget and preferences over bundles of goods. At CE, each buyer receives an optimal bundle, and the market clears.
The CE in such markets is captured by the Eisenberg-Gale convex program \citep{eisenberg1959consensus}, which maximizes Nash social welfare---defined as the geometric mean of agents' utilities. Later, \citet{Shmyrev09} proposed an alternative convex formulation for this problem. \citet{DGJMVY17} establish duality connections between these and other related convex programs, and also provide a convex formulation for the case of quasi-linear utilities.
\citet{DevanurPSV08} give a combinatorial polynomial-time algorithm for computing CE, which was subsequently improved to strongly polynomial-time algorithms in \citet{Orlin10} and \citet{Vegh16}.
Polynomial-time computability extends to more general classes such as homogeneous and weak gross substitutes (WGS) utility functions~\citep{Eisenberg61,codenotti2005polynomial,BeiGH19}, but beyond these, computing a CE becomes essentially PPAD-hard, particularly when the CE set is non-convex~\citep{ChenT09,ChenDDT09,codenotti2006leontief,chen2017complexity,GargMVY17,Rubinstein18,DeligkasFHM24}.

\paragraph{Stackelberg equilibrium.}
Stackelberg games~\citep{von1934marktform} model strategic interactions between a leader, who commits to a strategy first, and one or more followers, who observe the leader’s choice and best respond to it. These games naturally capture many real-world scenarios, including pricing, security, and finance, and can be viewed as extensions of the classic minimax optimization problem~\citep{neumann1928} to the more complex setting of general-sum games.
Computing Stackelberg equilibria is challenging due to the bi-level structure of the problem: the leader must optimize its payoff while anticipating followers’ best responses. This formulation often results in non-convex, combinatorial optimization problems that are NP-hard in general~\citep{Roughgarden01,conitzer2006computing,KorzhykCP10}.
Several works have developed efficient algorithms for special cases or restricted domains; see e.g.,~\citep{vonstengel2010leadership,blum2019computing,GoktasZG22}. Stackelberg games provide a powerful framework for modeling leader-follower interactions, though computational tractability remains a key challenge.

\section{Approximation Algorithm}
\label{sec:approx-algo}

In this section, we give approximation algorithms for revenue-maximization in a data market.
We begin with a definition of $k$-submodularity, a key ingredient in our first result.

\begin{definition}[$k$-submodularity]
Let $\Pi_k(U)$ be the set of all $k$-tuples of disjoint subsets of $U$, i.e.,
$\Pi_k(U) \defeq \{(S_1, \ldots, S_k): S_i \subseteq U
    \text{ and } S_i \cap S_j = \emptyset \text{ for all } i \neq j\}$.
Given a finite nonempty set $U$, a function $f: \Pi_k(U) \rightarrow \mathbb{R}_{\ge 0}$ defined on $k$ disjoint subsets of $U$ is called \emph{$k$-submodular} if for all tuples $(S_1, \dots, S_k) \in \Pi_k(U)$ and $(T_1, \dots, T_k) \in \Pi_k(U)$ of disjoint subsets of $U$, we have
\[ f(S_1, \dots, S_k) + f(T_1, \dots, T_k) \ge f((S_1, \dots, S_k) \sqcup (T_1, \dots, T_k)) + f((S_1, \dots, T_k) \sqcap (T_1, \dots, T_k)), \]
where we define
\begin{align*}
(S_1, \dots, T_k) \sqcap (T_1, \dots, T_k) &= (S_1\cap T_1, \dots, S_k \cap T_k)
\\ (S_1, \dots, T_k) \sqcup (T_1, \dots, T_k) &= \Big((S_1\cup T_1)\setminus \bigcup_{j\neq 1}(S_j\cup T_j),\dots, (S_k\cup T_k)\setminus \bigcup_{j\neq k}(S_j\cup T_j) \Big)\,.
\end{align*}
\end{definition}

We now analyze the greedy algorithm: initialize $\p=\vecZero$, and iteratively
raise the price of each dataset as much as possible---up to the point
where further increases would cause the revenue to decline.
We show that this algorithm is 2-approximate.

\begin{lemma}
\label{thm:rev-submod}
In a data market instance, let $r(S_1, \ldots, S_n)$ be the revenue when
for each buyer $i \in [n]$, each dataset $j \in S_i$ has price $\v_{i,j}$,
and each dataset outside $S_1 \cup \ldots \cup S_n$ has price 0.
Then the revenue function $r(S_1, \dots, S_n)$ is monotone and $n$-submodular.
\end{lemma}
\begin{proof}
It is clear that $r(S_1, \dots, S_n)$ is monotone, since setting a price to
a positive value yields (weakly) higher revenue than setting it to zero.
Next, we prove the $n$-submodularity.
Since $r(\cdot)$ is monotone, it suffices to prove that
$S_i\cap T_{i'} = T_{i} \cap S_{i'} = \emptyset$ for any $i\neq i'$.
Otherwise, if $S_i\cap T_{i'} \neq \emptyset$ we can reduce $S_i$ to $S_i \setminus T_{i'}$,
which does not change the right-hand side and weakly reduces the left-hand side.
Using a similar argument as in the proof of the equivalence between
the set-function and diminishing-returns definitions of submodularity,
it suffices to prove the following condition
\begin{align*}
r(S_1, \dots, S_i\cup \{j\}, \dots, S_n) - r(S_1, \dots, S_n) \le
r(T_1, \dots, T_i\cup \{j\}, \dots, T_n) - r(T_1, \dots, T_n)
\end{align*}
for any $T_{i'}\subseteq S_{i'}, i'\in [n]$ and $j\in M\setminus S_i$.

Comparing the prices corresponding to $(S_1,\dots, S_n)$ (denoted by $\p^S$)
and $(T_1, \dots, T_n)$ (denoted by $\p^T$),
we can observe that every buyer has more remaining budget under $\p^T$,
since the prices are either set the same as in $\p^S$ or zero.
Therefore, changing the price of $\D_j$ from $0$ to $\v_{i, j}$
results in a weakly higher revenue increase under $\p^T$ than under $\p^S$,
which leads to the $n$-submodularity.
\end{proof}

\begin{lemma}
\label{lem:2_approx_linear}
For maximizing revenue in a data market, the greedy algorithm is $2$-approximate,
and this is nearly the best possible approximation ratio for the greedy algorithm.
\end{lemma}
\begin{proof}
According to \cite{ward2016maximizing}, there exists a greedy algorithm
that achieves an approximation ratio of $2$ for a monotone $n$-submodular function.
The algorithm first specifies an arbitrary order $\sigma$ of elements of $U$.
At round $i$, it assigns the $i$-th element to the subset with the largest marginal increase.
\cite{10.5555/2884435.2884465} then improved the approximation ratio to
$2-1/n$ by using a randomized variant of the greedy algorithm.
Since the revenue function is monotone and $n$-submodular by \cref{thm:rev-submod},
the above results apply to our setting and show that the greedy algorithm is $2$-approximate.

It is worth noting that the deterministic greedy algorithm
can be adapted to an \emph{online} setting, where datasets arrive sequentially,
and still achieves a competitive ratio of $2$.

Due to the last point, one may wonder whether it is possible for the
deterministic algorithm to achieve the optimal revenue for \emph{some} sequence $\sigma$.
However, it does not always hold, and the following example shows that
the deterministic greedy can be suboptimal over all possible sequences $\sigma$.

\begin{example}[Deterministic greedy is suboptimal over all $\sigma$]
\label{exampe:deterministic_greedy_suboptimal}
Consider a marketplace with two buyers and three datasets.
Both buyers have budgets of $1$.
The values are set as follows: $\v_{1,1} = \v_{1, 2} = 0.2$, $\v_{1,3} =0$ and $\v_{2,1} = \v_{2,2} = 0.6$, $\v_{2,3} = 0.5$, as illustrated in \cref{fig:greedy_suboptimal}.
The optimal price vector is $\p = (0.2, 0.2, 0.5)$, which yields a revenue of $r(\p) = 0.4 + 0.9 = 1.3$.
However, the greedy algorithm will prioritize setting the price of the first two datasets to $0.6$ to fully use up buyer $2$'s budget.
As a result, the revenue gained from buyer 1 can at most be $0.2$, leading to a total revenue of at most $1 + 0.2 = 1.2$.
Therefore, the deterministic greedy will always return a suboptimal pricing.
\begin{figure}[ht]
\centering
\begin{tikzpicture}[
    agent/.style={circle, draw, fill=blue!50!black!50, text=bgColor, minimum size=0.2cm, font=\small},
    job/.style={rectangle, draw, fill=green!50!black!50, text=bgColor, minimum size=0.2cm, font=\small},
    edge/.style={thick},
]

\node[agent] (a1) at (0, 1) {$a_1$};
\node[agent] (a2) at (0, 0) {$a_2$};

\node[job] (j1) at (2, 2) {$j_1$};
\node[job] (j2) at (2, 0.5) {$j_2$};
\node[job] (j3) at (2, -1) {$j_3$};
\draw[edge] (a1) -- (j1) node[midway, above] {\scriptsize 0.2};
\draw[edge] (a2) -- (j1) node[midway, above] {\scriptsize 0.6};
\draw[edge] (a1) -- (j2) node[midway, above] {\scriptsize 0.2};
\draw[edge] (a2) -- (j2) node[midway, above] {\scriptsize 0.6};
\draw[edge] (a2) -- (j3) node[midway, above] {\scriptsize 0.5};

\node at (2.75, 2) {\scriptsize \textcolor{red!75!textHeavy}{$\$$0.2}};
\node at (2.75, 0.5) {\scriptsize \textcolor{red!75!textHeavy}{$\$$0.2}};
\node at (2.75, -1) {\scriptsize \textcolor{red!75!textHeavy}{$\$$0.5}};
\node at (-1, 1) {\small $1$};
\node at (-1, 0) {\small $1$};
\end{tikzpicture}

\caption{Example where the deterministic greedy algorithm is suboptimal over all $\sigma$}
\label{fig:greedy_suboptimal}
\end{figure}
\end{example}

Finally, we provide an example to show that $2$ is nearly the best ratio for the deterministic greedy algorithm.
Consider a data market instance with $n$ buyers and two datasets. The parameters are set as follows:
$b_i = 1+\eps$ and $\v_{i, 1} = 1+\eps, \v_{i, 2} = 1$ for any $i \in [n-1]$;
$b_n = n+1$ and $\v_{n+1, 1} = n, \v_{n+1, 2} = 1+\eps$,
where $\eps$ is a sufficiently small positive constant.
Suppose the algorithm sets the prices in the order of $\D_1, \D_2$.
For $p_1$, the marginal increase of the revenue $r$ is $n\cdot(1+\eps)$ when $p_1 = 1+\eps$,
while the increase is $n$ by setting $p_1 = n$.
Afterward, since the budgets of all the first $n-1$ buyers have been used up,
setting $p_2 = 1+\eps$ increases the revenue the most.
Therefore, the price vector output by the greedy algorithm is $(1+\eps, 1+\eps)$,
leading to a revenue of $(n+1)(1+\eps)$.
However, the optimal prices under this instance is $p_1^* = n$ and $p_2^* = 1$, leading to a revenue of $2n$.
The approximation ratio is $\frac{2n}{(n+1)(1+\eps)} \approx 2 - 2/(n+1)$ for small $\eps$.
\end{proof}

\begin{theorem}
The problem of revenue maximization in a data market
has a randomized $(1-1/e)^{-1}$-approximation algorithm.
\end{theorem}
\begin{proof}
Let $U$ be the set of datasets $M$.
Define the ground set $U'$ by creating $n$ copies of every element of $U$.
Denote by $x^{(i)}$ the $i$-th copy of element $x$.
Then we introduce a function $\hat{r}$ defined over $2^{U'}$ as follows:
\begin{align*}
\hat{r}(S) = \sum_{i=1}^n \hat{r}_i(S) = \sum_{i=1}^n \min\left(b_i, \sum_{j=1}^m \sum_{\ell=1}^n \v_{i, \ell, j} \cdot \mathbbold{1}[j^{(\ell)}\in S]\right), \quad
\text{where }
\v_{i, \ell, j} = \begin{cases}
\v_{\ell, j}, & \text{if } \v_{\ell, j}\le \v_{i, j}\\
0, & \text{otherwise}
\end{cases}
\end{align*}
Since $\v_{i, \ell, j}$ are all non-negative, then $\hat{r}(\cdot)$ is monotone.
Next, we prove that $\hat{r}(\cdot)$ is submodular.
Since the summation does not destroy submodularity, it suffices to prove every inner term of $\hat{r}_i(S)$ is submodular, equivalently, for every $T\subseteq S$ and $e =j'^{(\ell')}\in U\setminus S$,  $\hat{r}_i(S\cup \{e\}) - \hat{r}(S) \le \hat{r}(T\cup \{e\}) - \hat{r}(T)$.
Notice that,
\begin{align*}
&\hat{r}_i(S\cup \{e\}) - \hat{r}_i(S) \\
& = \min\Big(b_i, \sum_{j=1}^m \sum_{\ell=1}^n \v_{i, \ell, j} \cdot \mathbbold{1}[j^{(\ell)}\in S\cup \{e\}]\Big) -
\min\Big(b_i, \sum_{j=1}^m \sum_{\ell=1}^n \v_{i, \ell, j} \cdot \mathbbold{1}[j^{(\ell)}\in S]\Big)\\
& = \min\Big(b_i, \sum_{j=1}^m \sum_{\ell=1}^n \v_{i, \ell, j} \cdot \mathbbold{1}[j^{(\ell)}\in S]+ \v_{i, \ell', j'}\Big) -
\min\Big(b_i, \sum_{j=1}^m \sum_{\ell=1}^n \v_{i, \ell, j} \cdot \mathbbold{1}[j^{(\ell)}\in S]\Big) \\
& \le  \min\Big(b_i, \sum_{j=1}^m \sum_{\ell=1}^n \v_{i, \ell, j} \cdot \mathbbold{1}[j^{(\ell)}\in T]+ \v_{i, \ell', j'}\Big) -
\min\Big(b_i, \sum_{j=1}^m \sum_{\ell=1}^n \v_{i, \ell, j} \cdot \mathbbold{1}[j^{(\ell)}\in T]\Big)\\
& = \hat{r}_i(T\cup \{e\}) - \hat{r}_i(T),
\end{align*}
where the inequality is due to the submodularity of the min function and $T\subseteq S$.
Therefore, the function $\hat{r}_i$ is submodular.

Next, consider the partition matroid $\mathcal{I}$: $\mathcal{I} = \{S: \abs{S \cap \{j^{(i)}\}_{i=1}^n} \le 1 \text{ for any }j\in M\}$, where every set in $\mathcal{I}$ contains at most one copy of every element of $M$.
We show a bijective mapping between $\hat{r}(\cdot)$ defined on each set of the partition matroid and $r(\cdot)$ defined on an $n$-disjoint tuple.
For every set $S$, we define $(S_1, \dots, S_n)$ as follows: $j\in S_i$ only if $j^{(i)}\in S$.
Since $S$ belongs to the partition matroid, then $S_1,\dots, S_n$ are disjoint.
In addition, it can be checked that $\hat{r}(S)$ equals the corresponding value of $r(S_1, \dots, S_n)$.

Therefore, the problem of maximizing $r(S_1, \dots, S_n)$ reduces to the problem of maximizing $\hat{r}(S)$ subject to the partition matroid $S\in \mathcal{I}$.
By \cite{calinescu2011maximizing}, the continuous greedy algorithm can give a randomized solution that achieves an approximation ratio of $(1-1/e)^{-1}$ in expectation.
By randomly selecting the prices corresponding to each possible $S$,
we achieve an expected approximation ratio of $(1-1/e)^{-1}$ for the revenue.
\end{proof}

\section{Inapproximability Results}
\label{sec:inapprox}

We start with a simple reduction to show that finding the optimal prices is \NP-hard even for two buyers and under equal budgets.

\begin{lemma}
Finding the optimal prices is \NP-hard even under equal budgets.
\end{lemma}
\begin{proof}
We show a reduction from the partition problem, where one is given a set of positive numbers $\{a_1, \ldots, a_m\}$ with $\sum_{j=1}^m a_j = 2b$ and the goal is to determine whether the set can be partitioned in $(A, B)$ such that $\sum_{j\in A}a_j = \sum_{j\in B}a_i$.
We construct a data market instance with two buyers and $m+1$ datasets.
The two buyers have the same budget of $3b$.
For every $\D_j$, $j\in [m]$, the first buyer has a value of $\v_{1, j} = a_j$ while the second buyer has a value of $\v_{2, j} = 2a_j$.
For the last dataset $\D_{m+1}$, the first buyer has a value of $\v_{1, m+1} = 2b$ while the second buyer has a value of $\v_{2, m+1} = 0$.

\begin{table}[ht]
    \centering
    \begin{tabular}{c|cccc}
       $j$  & 1 & $\dots$ & $m$ & $m+1$ \\
       \hline
      $\v_{1, j}$ (with budget of $3b$)  & $a_1$  & $\dots$ & $a_m$ & $2b$\\
      $\v_{2, j}$ (with budget of $3b$)  & $2a_1$ & $\dots$ & $2a_m$ & $0$  \\
    \end{tabular}
\end{table}
When the partition instance is a YES instance, suppose the index set can be partitioned into $A$ and $B$ such that $\sum_{j\in A} a_j = \sum_{j\in B} a_j = b$.
We set the prices as follows: for every dataset $\D_j$ with index $j\in A$, set the price as $a_j$; for every dataset $\D_j$ with $j\in B$, set the price as $2a_j$; set the price of the last one as $2b$.
Then the revenue gained from the first buyer is equal to $\min(b_1, \sum_{j: \v_{1, j}\ge p_j} p_j) = \min(3b, \sum_{j\in A} p_j + 3b) = 3b$.
In addition, as $\v_{2, j} \ge p_j$ always holds for any $j\in [m]$, the revenue from the second buyer is equal to $b_2 = 3b$.
Therefore, the total revenue is equal to $3b+3b = 6b$.

When the partition instance is an NO instance, we now prove it is impossible to earn a revenue of $6b$.
Otherwise, assume there exists a price vector $\p$ that achieves revenue of $6b$.
According to \cref{thm:rdisc}, there is an optimal price vector such that $p_j \in \{a_j, 2a_j\}$ for any $j\in [m]$ and $p_{m+1} = 2b$.
Let $A$ be the set of indices where $p_j = a_j$ and $B$ be the set of indices where $p_j = 2a_j$.
We can observe that buyer 1 is only willing to buy datasets in $A$, and the last one, while buyer 2 is willing to buy any dataset in the first $m$ ones.
As the total revenue is equal to $6b$, we then have
\begin{align*}
3b = b_1 \le \sum_{j\in A} p_j + 2b= \sum_{j\in A} a_j + 2b,\quad
3b = b_2 \le \sum_{j\in [m]} p_j = \sum_{j\in A}a_j + \sum_{j\in B} 2a_j = 2b + \sum_{j\in B} a_j,
\end{align*}
leading to $\sum_{j\in A} a_j, \sum_{j\in B} \ge b$, which contradicts the fact that the partition instance is a NO instance.
\end{proof}

Next, we show that the problem has a constant inapproximability.
The hardness is by a reduction from the vertex cover problem on regular graphs.

\thmLinearAPX*
\begin{proof}
We present a reduction from the Minimum Vertex Cover problem, whose goal is to find the smallest possible set of vertices such that every edge in the graph is incident to at least one vertex in the set.
Due to \cite{dinur2007pcp} and the PCP theorem, the minimum vertex cover is hard to approximate within a constant factor, even in graphs with constant bounded degrees.
Formally,
\begin{lemma}[Theorem 2.10 of \cite{schoenebeck2020influence,dinur2007pcp}]
\label{thm:vertex_cover}
There exists a universal constant integer $d$ and a universal constant $\gamma \in (0, 1)$ such that,
given an integer $k = \frac{2}{3}n$ and an undirected graph $G=(V, E)$ where the degree of each vertex is bounded by $d$, it is \NP-hard to distinguish the following two cases:
\begin{itemize}
    \item YES: $G$ has a vertex cover of size $k$;
    \item NO: all vertex covers of $G$ has size at least $(1 + \gamma)k$.
\end{itemize}
\end{lemma}
Given a vertex cover instance $G=(V,E)$ constructed by \cref{thm:vertex_cover}, we now construct a data market instance as follows: there are $n$ datasets, where $n$ is the number of vertices of the given graph.
Each dataset corresponds to one vertex $v$ of the input graph.
There are two types of buyers in the data market: $m$ \emph{edge buyers} and $t$ \emph{normal buyers}.
Each edge buyer $e$ corresponds to an edge of the input graph.
Each edge buyer $e=(u, v)$ has value $B$ for each of the two datasets $u$ and $v$
and value $0$ for any other dataset.
Each normal buyer has the same value of $1$ for every dataset.
We set $B$ and $t$ sufficiently large such that $\eps\cdot t > m +n$ and $B > (1+\eps)\cdot t$.
Each edge buyer has a budget of $B$, and each normal buyer has a budget of $n$.
\begin{table}[ht]
    \centering
    \begin{tabular}{l|c|ccccccc}
       &  & 1 & $\dots$ & $u$ & $\dots$ &$v$ & $\dots$ & $n$  \\
       \hline
     Normal buyers & $\v_{n_1, j}$ (with budget of $n$)  & $1$  &  & & $\dots$ &  &  & $1$\\
     & $\dots$ \\
     &$\v_{n_t, j}$ (with budget of $n$)  & $1$   &  & & $\dots$ &  &  & $1$\\
     Edge buyers &  $\v_{(u, v), j}$ (with budget of $B$)  & 0  & $\dots$ & $B$ & $0$ & $B$ & $0$ & $0$\\
    \end{tabular}
\end{table}

To establish the gap in optimal revenue between the yes and no instances, we first present the following lemma, which characterizes the optimal prices.
\begin{lemma}
Suppose $\p^*$ is the optimal price vector.
Then there are exactly $k$ among the prices equal to $B$,
with the remaining $n-k$ prices equal to $1$, where $k$ is the size of the minimum vertex cover.
\end{lemma}

\begin{proof}
According to \cref{thm:rdisc}, there exists a unique optimal price vector $\bm{p}^*$ with $p_i^*\in \{1, B\}$.
We first construct prices according to the minimum vertex cover, denoted by $S$.
Let $k$ be the size of $S$.
We set $p_i^* = B$ if $i \in S$ and $b$ otherwise.
In that case, every edge buyer is able to fully consume her budget since there exists one of its endpoints whose price is set as $B$.
Hence, the total revenue from edge buyers is equal to $m\cdot B$.
Meanwhile, each normal buyer is only willing to buy the items for which the price is $1$.
Therefore, the total revenue from the normal buyers is equal to $t\cdot (n-k)$.
Combining it up, we have the total revenue under this price vector is equal to $m\cdot B + t\cdot (n-k)$.

Next, we prove that the revenue is suboptimal when the set of datasets with price $B$ does not form a minimum vertex cover.
If not, assume $\p$ is also an optimal price vector while $V^B= \{v: p_v = B\}$ is not a minimum vertex cover of $G$.
First, $V^B$ must be a vertex cover of $G$.
Otherwise, assume edge $e=(u, v)$ is not covered.
By changing the price vector into $\hat{\bm{p}}$ with $\hat{p_u} = B$ and $\hat{p}_w = p_w$ for any other $w\in V$, the increase in revenue will be at least
\[
B - (m-1)\cdot b - n - t  =  B - (m+n+t-1) >0,
\]
which violates the optimality of $\p$.
Then, since $V^B$ is not the minimum, then  $\abs{V^B} > k$ by the optimality of $S$.
Hence, the maximum revenue is at most $m\cdot B + t\cdot (n-k-1)$, which is less than the revenue achieved by $\p^B$ and violates the optimality of $\p$.
\end{proof}
Let $k^Y$ and $k^N$ respectively be the size of the minimum vertex cover and $\p^{\sf YES}$ and $\p^{\sf NO}$ be the corresponding optimal price vector of the YES and NO instance of \cref{thm:vertex_cover}.
Hence, $k^Y = \frac{2}3n$ and $k^N > \frac{2}3n\cdot (1+\gamma)$.
When the input vertex cover instance is YES, the maximum revenue is
\[
r(\p^Y) = m\cdot B + t\cdot (n - k^Y) = m\cdot B + t\cdot \frac{1}{3}n\,.
\]
On the other hand, when the input vertex cover instance is a NO instance, then the maximum revenue is
\[
r(\p^N) = m\cdot B + t\cdot (n - k^N) < m\cdot B + t\cdot \frac{1-2\gamma}{3} \cdot n
\]
Therefore, the inapproximability ratio is at least
\begin{align*}
\frac{r(\p^Y)}{r(\p^N)} & \ge \frac{m\cdot B + t\cdot \frac{1}{3}n}{m\cdot B + t\cdot \frac{1-2\gamma}{3} \cdot n} \ge \frac{\frac{d}{2}n\cdot B + t\cdot \frac{1}{3}n}{\frac{d}{2}n\cdot B + t\cdot \frac{1-2\gamma}{3} \cdot n}
= \frac{\frac{d}{2}\cdot B + t\cdot \frac{1}{3}}{\frac{d}{2}\cdot B + t\cdot \frac{1-2\gamma}{3}} \\
& = \frac{\frac{d}{2}\cdot t\cdot (1+\eps) + t\cdot \frac{1}{3}}{\frac{d}{2}\cdot t\cdot (1+\eps) + t\cdot \frac{1-2\gamma}{3}}
= 1 + \frac{4\gamma}{3d+ 2-4\gamma} - \eps',
\end{align*}
where $\eps' = O(\eps)$.
Therefore, for any constant $\eps > 0$, the problem of finding optimal prices has a constant inapproximability ratio of $1 + \frac{4\gamma}{3d+ 2-4\gamma} - \eps$, where $\gamma$ and $d$ are universe constants provided by \cref{thm:vertex_cover}.
As a corollary, finding the prices is \APX-hard.
\end{proof}

\section{Supply-Side Market Clearing}
\label{sec:market-clear}

In addition to maximizing revenue,
we would also like to get \emph{supply-side market clearance},
i.e., every dataset is completely allocated to some buyer.
This ensures that no data goes unutilized.
In this section, we show that market clearance is compatible with revenue maximization.

We start by proving a sufficient condition for a price vector to allow market clearance.

For a price vector $\pvec$ and buyer $i$, define her \emph{desire} to be the amount of money
she needs to buy her favorite bundle. Equivalently, $d_i(\pvec)$ is the
maximum revenue we could have obtained from buyer $i$ if her budget was infinite.
Using \cref{thm:rdisc}, we get that
\[ d_i(\pvec) = \sum_{j \in [m]: \v_{i,j} \ge p_j} p_j. \]
Hence, $r_i(\pvec) = \min(b_i, d_i(\pvec))$.
A buyer is said to be \emph{satisfied} by $\pvec$ if $d_i(\pvec) \le b_i$,
i.e., she is not limited by her budget.

\begin{lemma}
\label{thm:clearable-implies-mc}
Call a price vector $\pvec$ \emph{clearable} if for every dataset $j \in [m]$,
either $p_j = 0$, or there exists a satisfied buyer $i \in [n]$ such that $\v_{i,j} \ge p_j$.
For any clearable price vector $\pvec$, there exists
a market-clearing allocation $\xvec \in [0, 1]^{n \times m}$, i.e.,
for every dataset $j \in [m]$, there exists a buyer $i \in [n]$ such that $x_{i,j} = 1$.
\end{lemma}
\begin{proof}
Fix a dataset $j$. If $p_j = 0$, then every buyer can be given the entire dataset.
Now assume $p_j > 0$. Then there exists a satisfied buyer $i$ such that $\v_{i,j} \ge p_j$.
This buyer would like to buy dataset $j$, and since her budget does not limit her, she will buy it entirely.
\end{proof}

We now show that for every price vector, there exists
another clearable price vector having greater or equal revenue.
This would prove that market clearance is compatible with revenue maximization.

\begin{lemma}
\label{thm:clearabilize}
For any price vector $\pvechat$, we can compute
a clearable price vector $\pvec^*$ in $O(m^2n^3)$ time such that
for each buyer $i \in [n]$, we have $r_i(\pvec^*) \ge r_i(\pvechat)$.
\end{lemma}
\begin{proof}
We will show that any unclearable price vector $\pvec$ can be modified so that
a potential function of finite range decreases,
$r_i(\pvec)$ doesn't decrease for any buyer $i$,
and $p_j$ doesn't increase for any dataset $j$.
By iteratively applying this modification starting from $\pvechat$,
we will eventually find the required price vector $\pvec^*$.

For any price vector $\pvec$, define
\begin{align*}
\phi_1(\pvec) &\defeq \sum_{j=1}^m \left(\boolOne(p_j > 0) + \sum_{i=1}^n \boolOne(\v_{i,j} < p_j)\right),
\\ \phi_2(\pvec) &\defeq \sum_{i=1}^n \boolOne(d_i(\pvec) > b_i).
\end{align*}
Let $\phi(\pvec) \defeq (n+1)\phi_1(\pvec) + \phi_2(\pvec)$ be our potential function.
Then $0 \le \phi_1(\pvec) \le m(n+1)$ and $0 \le \phi_2(\pvec) \le n$,
so $0 \le \phi(\pvec) < (m+1)(n+1)^2$.

Suppose $\pvec$ is unclearable.
Let $C \defeq \{i \in [n]: d_i(\pvec) > b_i\}$ be the \emph{constrained} (i.e., unsatisfied) buyers,
and for each dataset $j \in [m]$, let
$U_j \defeq \{i \in [n]: \v_{i,j} < p_j\}$ be the buyers \emph{uninterested} in $j$ at price $p_j$.
Since $\pvec$ is unclearable, there exists a dataset $j \in [m]$
such that $p_j > 0$ and $[n] = U_j \cup C$.

We will reduce the price of dataset $j$ such that
some buyer becomes satisfied or the price reduces to 0.
Formally, let $p'_k \defeq p_k$ for all $k \in [m] \setminus \{j\}$,
and let $p'_j \defeq \max(0, \beta)$, where
\[ \beta \defeq \max_{i \in C \setminus U_j} (p_j - d_i(\pvec) + b_i). \]
We can compute $\pvec'$ in $O(mn)$ time.
For all $i \in C$, we have $d_i(\pvec) > b_i$, so $\beta < p_j$.
Since $p_j > 0$, we get $0 \le p'_j < p_j$.

For all $i \in C \setminus U_j$, we have $d_i(\pvec') = d_i(\pvec) - p_j + p'_j$
and $\beta \ge p_j - d_i(\pvec) + b_i$. Add these two inequalities to get
$d_i(\pvec') + \beta \ge p'_j + b_i$, which implies $d_i(\pvec') \ge (p'_j - \beta) + b_i \ge b_i$.
Hence, for all $i \in C \setminus U_j$, we have $r_i(\pvec') = r_i(\pvec) = b_i$.
For all $i \in U_j$, we have $d_i(\pvec') \ge d_i(\pvec)$, so $r_i(\pvec') \ge r_i(\pvec)$.

If $\phi_1(\pvec') < \phi_1(\pvec)$, then $\phi(\pvec') < \phi(\pvec)$.
Now suppose $\phi_1(\pvec') = \phi_1(\pvec)$. Then $p'_j > 0$,
and $d_i(\pvec') = d_i(\pvec)$ for all $i \in U_j$.
Let $C' \defeq \{i \in [n]: d_i(\pvec') > b_i\}$.
We want to show that $|C'| < |C|$, which would imply $\phi_2(\pvec') < \phi_2(\pvec)$.

If $i \in C' \cap U_j$, then $i \in C \cap U_j$.
If $i \in C' \setminus U_j$, then $i \in C$, since $[n] = C \cup U_j$.
Hence, $C' \subseteq C$. Since $p'_j > 0$, we have $p_j = \beta$.
Hence, $p'_j = p_j - d_i(\pvec) + b_i$ for some $i \in C \setminus U_j$.
Then $d_i(\pvec') = d_i(\pvec) - p_j + p'_j = b_i$. Hence, $i \not\in C'$.
Hence, $|C'| < |C|$. Therefore, $\phi_2(\pvec') < \phi_2(\pvec)$,
which implies $\phi(\pvec') < \phi(\pvec)$.

By starting from unclearable prices $\pvec$, we obtained another price vector $\pvec'$
such that $\phi(\pvec') < \phi(\pvec)$, $r_i(\pvec') \ge r_i(\pvec)$ for all $i \in [n]$,
and $p'_j \le p_j$ for all $j \in [m]$.
By applying this process iteratively on $\phat$,
we eventually obtain clearable prices $p^*$ such that
$r_i(p^*) \ge r_i(\phat)$ for all $i \in [n]$ and $p^*_j \le \phat_j$ for all $j \in [m]$.

The total number of iterations is upper-bounded by
the number of different values the potential function $\phi$ can take,
and $\phi(\pvec)$ can take at most $(m+1)(n+1)^2$ different values.
Moreover, it takes $O(mn)$ time to check if a price vector is clearable,
and $O(mn)$ time to run a single iteration.
Hence, the total running time of the algorithm is $O(m^2n^3)$.
\end{proof}

\newpage
\appendix
\crefalias{section}{appendix}
\crefalias{subsection}{appendix}
\crefalias{subsubsection}{appendix}

\section{Form of Accuracy Function \texorpdfstring{$a_i(\cdot)$}{a\_i(.)}}
\label{app:form_of_accuracy}

For the purpose of our proofs we work under the assumption that $x_{ij}$ is integral for all $i,j$.
\begin{lemma}
    \label{lem:utility}
     We have $\bm a_i(\bm x_i) = \mathbb{E}[\mathrm{Pre}(\theta_i \mid S(\bm x_i))] - \mathrm{Pre}(\theta_i) =  \sum_j \tau_{i,j} x_{i,j}$.
\end{lemma}

\begin{proof}
 Let $s_{i,j}(1), s_{i,j}(2), \dots , s_{i,j}(x_{i,j})$ denote the $x_{i,j}$ signals associated with the $x_{i,j}$ data records sampled from $\D_j$. We first show that $\theta_i \mid S(\bm x_i) \sim N \Big (\frac{\sum_{j \in M}\sum_{\ell \in [x_{i,j}]}s_{i,j}(\ell) \tau_{i,j} + \tau_i\mu_i }{\sum_{j \in M} \tau_{i,j}x_{i,j} + \tau_i}, (\tau_i + \sum_{j \in M} \tau_{i,j}x_{i,j})^{-1} \Big)$.

 By Bayes' theorem, we have:
 \begin{align*}
    \Pr[\theta_i \mid S(\bm x_i)] = \frac{\Pr[S( \bm x_i)\mid \theta_i] \Pr[\theta_i]}{\Pr[S(\bm x_i)]}
 \end{align*}

Observe that $\Pr[\theta_i] \propto \exp \big(\frac{(\theta_i -\mu_i)^2}{2} \big)$. Further, since each $s_{i,j}(\ell)$ is an independent sample from the underlying distribution of $\D_j$, we have $\Pr[S(\bm x_i) \mid \theta_i] = \prod_{j \in M} \prod_{\ell \in [x_{i,j}]} \Pr[\eta_{i,j} = s_{i,j}(\ell) - \theta_i] = \prod_{j \in M} \prod_{\ell \in [x_{i,j}]} \exp(\tau_{i,j}(s_{i,j}(\ell)-\theta_i)^2/2)$. Now, observe that

\begin{align*}
    \Pr[\theta_i \mid S(\bm x_i)] &\propto {\Pr[S( \bm x_i)\mid \theta_i] \Pr[\theta_i]}\\
                                  &=\Big(\prod_{j \in M} \prod_{\ell \in [x_{i,j}]} \exp(\tau_{i,j}(s_{i,j}(\ell)-\theta_i)^2/2)\Big) \cdot \exp(\tau_i(\theta_i - \mu_i)^2/2)\\
                                  &=\exp \Big (\sum_{j \in M} \sum_{\ell \in [x_{i,j}] \tau_{i,j}} (s_{i,j}(\ell)-\theta_i)^2/2 + \tau_i(\theta_i - \mu_i)^2/2 \Big )\\
                                  &\propto \exp \Big( \Big( \theta^2_i(\tau_i + \sum_j \tau_{i,j}x_{i,j}) -2\theta_i \cdot \big( {\sum_{j \in M}\sum_{\ell \in [x_{i,j}]}s_{i,j}(\ell) \tau_{i,j} + \tau_i\mu_i } \big) \Big)/2\Big)\\
                                  &\propto \exp \Big( (\tau_i + \sum_j \tau_{i,j}x_{i,j}) \cdot \Big (\theta_i - \frac{\sum_{j \in M}\sum_{\ell \in [x_{i,j}]}s_{i,j}(\ell) \tau_{i,j} + \tau_i\mu_i }{\sum_{j \in M} \tau_{i,j}x_{i,j} + \tau_i} \Big)^2/2 \Big)
\end{align*}
implying that $\theta_i \mid S(\bm x_i) \sim N \Big (\frac{\sum_{j \in M}\sum_{\ell \in [x_{i,j}]}s_{i,j}(\ell) \tau_{i,j} + \tau_i\mu_i }{\sum_{j \in M} \tau_{i,j}x_{i,j} + \tau_i}, (\tau_i + \sum_{j \in M} \tau_{i,j}x_{i,j})^{-1} \Big)$. Therefore, we have $\mathrm{Var}(\theta_i \mid S(\bm x_i)) = (\tau_i + \sum_{j \in M}\tau_{i,j}x_{i,j})^{-1}$. Since, $\mathrm{Pre}(\theta_i \mid S(\bm x_i)) = \mathrm{Var}^{-1}(\theta_i \mid S(\bm x_i)) = \tau_i + \sum_j \tau_{i,j} x_{i,j} $ is non-stochastic, we have $\mathbb{E}[\mathrm{Pre}(\theta_i \mid S(\bm x_i)] =\mathrm{Var}^{-1}(\theta_i \mid S(\bm x_i)) = (\tau_i + \sum_{j \in M}\tau_{i,j} x_{i,j})$, implying that $a_i(\bm x_i) = \mathbb{E}[\mathrm{Pre}(\theta \mid S(\bm x_i)] - \mathrm{Pre}(\theta_i) =  \sum_j \tau_{i,j}x_{i,j}$. \qedhere
\end{proof}

\clearpage
\section{\texorpdfstring{$k$}{k}-submodular Maximization and Submodular Maximization subject to Partition Matroid}
\label{sec:k_sub_sub_partition}

Let $U =\{a, b\}$ be a ground set and $k=2$.
$f$ is a $k$-submodular functions defined on $(S_1, S_2)$ with $S_1, S_2\subseteq U$.
A possible transformation from $f$ is as follows:
Create $k$ copies $U^e = \{e^{(1)}, \dots, e^{(k)}\}$ for every element $e$ of $U$.
Let $\bar{U}$ be the set of copies.
Then we can extend the values of $f$ to a function $\hat{f} \colon 2^{\bar{U}} \to \mathbb{R}_{\ge 0}$ by setting $\hat{f}(S)$ as $f(S_1, \dots, S_k)$ with $S_i = \{e \in U: e^{(i)} \in S\}$ and define a partition matroid $\Ical = \{S: \abs{S\cap U^e} \le 1 \text{ for any }e \in U\}$.
However, an issue arises in the definition of $\hat{f}$: it can be undefined for subsets $S\subseteq 2^{\bar{U}}$ corresponding to non-disjoint $(S_1,\dots, S_k)$ since $f$ itself is only defined when the sets are disjoint.
If no further property beyond submodularity is required, we can still fix this issue by simply setting $\hat{f}(S)$ as zero if $S \notin \Ical$.
Unfortunately, such an extension is not always possible when monotonicity is also needed: there exists a monotone and $k$-submodular $f$ for which there exists no extension $\hat{f}$ satisfying both submodularity and monotonicity.

For example, \Cref{fig:value_of_f} defines the concrete values of $f$ and correspondingly we can define the values of $\hat{f}$ within the partition matroid, as discussed above.
It can be verified that $f$ is $2$-submodular.
\begin{figure}[ht]
\begin{minipage}[t]{0.49\textwidth}
\centering
   \begin{tabular}{c|cccc}
       \diagbox[height=1.8em]{$S_1$}{$S_2$} & $\emptyset$ & $\{a\}$ & $\{b\}$ & $\{a, b\}$\\
       \hline
       $\emptyset$  & 0 & 4 & 4 & 4\\
       $\{a\}$ & 1 & - & 4 & -\\
       $\{b\}$ & 1 & 5 & - & -\\
       $\{a, b\}$ & 1 & - & - & -\\
    \end{tabular}
\caption{The values of $f$}
\label{fig:value_of_f}
\end{minipage}
\hfill
\begin{minipage}[t]{0.49\textwidth}
   \begin{tabular}{c|ccc}
       \diagbox[width=6em,height=1.9em]{\scriptsize $S\cap U^a$}{\scriptsize $S\cap U^b$} & $\emptyset$ & $\{b^{(1)}\}$ & $\{b^{(2)}\}$ \\
       \hline
       $\emptyset$  & 0 & 1 & 4 \\
       $\{a^{(1)}\}$ & 1 & 1  & 4 \\
       $\{a^{(2)}\}$ & 4 & 5 & 4 \\
       $\{a^{(1)}, a^{(2)}\}$ & {\color{textBlue}\xmark} & {\color{textBlue} $5$} & {\color{textBlue} $4$} \\
    \end{tabular}
\caption{The values of $\hat{f}$}
\label{fig:value_hat_f}
\end{minipage}
\end{figure}

However, there does not exist a monotone and submodular function $\hat{f}$ defined over $2^{\bar{U}}$.
Otherwise, by submodularity and monotonicity of $\hat{f}$,
\begin{align*}
&4= \hat{f}(\{a^{(1)}, b^{(2)}\}) \le \hat{f}(\{a^{(1)}, a^{(2)}, b^{(2)}\}) \le \hat{f}(\{a^{(1)}, b^{(2)}\}) + \hat{f}(\{a^{(2)}, b^{(2)}\}) - \hat{f}(\{b^{(2)}\}) = 4,\\
&5= \hat{f}(\{a^{(2)}, b^{(1)}\})
\le \hat{f}(\{a^{(1)}, a^{(2)}, b^{(1)}\}) \le \hat{f}(\{a^{(1)}, b^{(1)}\}) + \hat{f}(\{a^{(2)}, b^{(1)}\}) - \hat{f}(\{b^{(1)}\}) = 5\,.
\end{align*}
On the one hand, by  monotonicity, we have $\hat{f}(\{a^{(1)}, a^{(2)}\}) \le \hat{f}(\{a^{(1)}, a^{(2)}, b^{(2)}\}) = 4$.
On the other hand, by submodularity, $\hat{f}(\{a^{(1)}, a^{(2)}\}) \ge \hat{f}(\{a^{(1)}, a^{(2)}, b^{(1)}\}) + \hat{f}(\{a^{(1)}\}) -\hat{f}(\{a^{(1)}, b^{(1)}\}) = 5 + 1 - 1 =5$, which contradicts to the previous requirement.
Therefore, the extension $\hat{f}$ does not exist.

\clearpage

\phantomsection
\addcontentsline{toc}{section}{References}

\newcommand{\etalchar}[1]{$^{#1}$}


\begin{thebibliography}{FMMO24}

\bibitem[ACGM25]{akrami2025theoretical}
Hannaneh Akrami, Bhaskar~Ray Chaudhury, Jugal Garg, and Aniket Murhekar.
\newblock On the theoretical foundations of data exchange economies.
\newblock In {\em ACM Conf.\ Economics and Computation (EC)}, pages 444--444,
  2025.
\newblock \href {https://doi.org/10.1145/3736252.3742566}
  {\path{doi:10.1145/3736252.3742566}}.

\bibitem[{Acu}22]{AcumenDM}
{Acumen Research}.
\newblock Big data market size: Global industry, share, analysis, trends and
  forecast 2022 - 2030, 2022.
\newblock URL:
  \url{https://www.acumenresearchandconsulting.com/big-data-market}.

\bibitem[AD54]{arrow1954existence}
Kenneth~J Arrow and Gerard Debreu.
\newblock Existence of an equilibrium for a competitive economy.
\newblock {\em Econometrica}, 22(3):265--290, 1954.
\newblock \href {https://doi.org/10.2307/1907353} {\path{doi:10.2307/1907353}}.

\bibitem[ADHR24]{AgarwalDHR20}
Anish Agarwal, Munther~A. Dahleh, Thibaut Horel, and Maryann Rui.
\newblock Towards data auctions with externalities.
\newblock {\em Games Econ. Behav.}, 148:323--356, 2024.
\newblock \href {https://doi.org/10.1016/j.geb.2024.09.008}
  {\path{doi:10.1016/j.geb.2024.09.008}}.

\bibitem[ADS19]{AgarwalDS19}
Anish Agarwal, Munther~A. Dahleh, and Tuhin Sarkar.
\newblock A marketplace for data: An algorithmic solution.
\newblock In {\em ACM Conf.\ Economics and Computation (EC)}, pages 701--726.
  {ACM}, 2019.
\newblock \href {https://doi.org/10.1145/3328526.3329589}
  {\path{doi:10.1145/3328526.3329589}}.

\bibitem[AFM{\etalchar{+}}23]{acemoglu2023good}
Daron Acemoglu, Alireza Fallah, Ali Makhdoumi, Azarakhsh Malekian, and Asuman
  Ozdaglar.
\newblock How good are privacy guarantees? platform architecture and violation
  of user privacy.
\newblock Technical report, National Bureau of Economic Research, 2023.

\bibitem[AGR11]{azar2011submodular}
Yossi Azar, Iftah Gamzu, and Ran Roth.
\newblock Submodular max-{SAT}.
\newblock In {\em European Symp.\ Algorithms (ESA)}, pages 323--334. Springer,
  2011.
\newblock \href {https://doi.org/10.1007/978-3-642-23719-5_28}
  {\path{doi:10.1007/978-3-642-23719-5_28}}.

\bibitem[AP86]{admati1986monopolistic}
Anat~R Admati and Paul Pfleiderer.
\newblock A monopolistic market for information.
\newblock {\em Journal of Economic Theory}, 39(2):400--438, 1986.
\newblock \href {https://doi.org/10.1016/0022-0531(86)90052-9}
  {\path{doi:10.1016/0022-0531(86)90052-9}}.

\bibitem[AP90]{admati1990direct}
Anat~R Admati and Paul Pfleiderer.
\newblock Direct and indirect sale of information.
\newblock {\em Econometrica: Journal of the Econometric Society}, pages
  901--928, 1990.
\newblock \href {https://doi.org/10.2307/2938355} {\path{doi:10.2307/2938355}}.

\bibitem[AS04]{ageev2004pipage}
Alexander~A Ageev and Maxim~I Sviridenko.
\newblock Pipage rounding: A new method of constructing algorithms with proven
  performance guarantee.
\newblock {\em Journal of Combinatorial Optimization}, 8(3):307--328, 2004.
\newblock \href {https://doi.org/10.1023/B:JOCO.0000038913.96607.c2}
  {\path{doi:10.1023/B:JOCO.0000038913.96607.c2}}.

\bibitem[BBG22]{bergemann2022economics}
Dirk Bergemann, Alessandro Bonatti, and Tan Gan.
\newblock The economics of social data.
\newblock {\em The RAND Journal of Economics}, 53(2):263--296, 2022.
\newblock \href {https://doi.org/10.1111/1756-2171.12407}
  {\path{doi:10.1111/1756-2171.12407}}.

\bibitem[BBS18]{bergemann2018design}
Dirk Bergemann, Alessandro Bonatti, and Alex Smolin.
\newblock The design and price of information.
\newblock {\em American economic review}, 108(1):1--48, 2018.
\newblock \href {https://doi.org/10.1257/aer.20161079}
  {\path{doi:10.1257/aer.20161079}}.

\bibitem[BGH19]{BeiGH19}
Xiaohui Bei, Jugal Garg, and Martin Hoefer.
\newblock Ascending-price algorithms for unknown markets.
\newblock {\em ACM Trans.\ Algorithms}, 2019.
\newblock \href {https://doi.org/10.1145/3319394} {\path{doi:10.1145/3319394}}.

\bibitem[BGI{\etalchar{+}}24]{BhaskaraGIKMS24}
Aditya Bhaskara, Sreenivas Gollapudi, Sungjin Im, Kostas Kollias, Kamesh
  Munagala, and Govind~S. Sankar.
\newblock Data exchange markets via utility balancing.
\newblock In {\em World Wide Web Conf.\ (WWW)}, pages 57--65. {ACM}, 2024.
\newblock \href {https://doi.org/10.1145/3589334.3645364}
  {\path{doi:10.1145/3589334.3645364}}.

\bibitem[BHS19]{blum2019computing}
A.~Blum, N.~Haghtalab, and S.~Seddighin.
\newblock Computing stackelberg equilibria of large general-sum games.
\newblock In {\em Algorithmic Game Theory (SAGT)}, pages 103--114. Springer,
  2019.

\bibitem[BK20]{bhaskar2018partial}
Umang Bhaskar and Gunjan Kumar.
\newblock Partial function extension with applications to learning and property
  testing.
\newblock In {\em Intl.\ Symp.\ Algorithms and Computation (ISAAC)}, volume
  181, 2020.
\newblock \href {https://doi.org/10.4230/LIPIcs.ISAAC.2020.46}
  {\path{doi:10.4230/LIPIcs.ISAAC.2020.46}}.

\bibitem[BKL12]{BabaioffKP12}
Moshe Babaioff, Robert Kleinberg, and Renato~Paes Leme.
\newblock Optimal mechanisms for selling information.
\newblock In {\em ACM Conf.\ Electronic Commerce (EC)}, page 92–109, 2012.
\newblock \href {https://doi.org/10.1145/2229012.2229024}
  {\path{doi:10.1145/2229012.2229024}}.

\bibitem[BV25]{baley2025data}
Isaac Baley and Laura~L. Veldkamp.
\newblock {\em The Data Economy: Tools and Applications}.
\newblock Princeton University Press, 2025.
\newblock \href {https://doi.org/10.1515/9780691256740}
  {\path{doi:10.1515/9780691256740}}.

\bibitem[CCPV11]{calinescu2011maximizing}
Gruia Calinescu, Chandra Chekuri, Martin Pal, and Jan Vondr{\'a}k.
\newblock Maximizing a monotone submodular function subject to a matroid
  constraint.
\newblock {\em SIAM Journal on Computing}, 40(6):1740--1766, 2011.
\newblock \href {https://doi.org/10.1137/080733991}
  {\path{doi:10.1137/080733991}}.

\bibitem[CDDT09]{ChenDDT09}
Xi~Chen, Decheng Dai, Ye~Du, and Shang{-}Hua Teng.
\newblock Settling the complexity of {A}rrow-{D}ebreu equilibria in markets
  with additively separable utilities.
\newblock In {\em Symp.\ Foundations of Computer Science (FOCS)}, pages
  273--282, 2009.
\newblock \href {https://doi.org/10.1109/FOCS.2009.29}
  {\path{doi:10.1109/FOCS.2009.29}}.

\bibitem[CDG{\etalchar{+}}17]{DGJMVY17}
Richard Cole, Nikhil~R. Devanur, Vasilis Gkatzelis, Kamal Jain, Tung Mai,
  Vijay~V. Vazirani, and Sadra Yazdanbod.
\newblock Convex program duality, {F}isher markets, and {N}ash social welfare.
\newblock In {\em ACM Conf.\ Economics and Computation (EC)}, pages 459--460.
  {ACM}, 2017.
\newblock \href {https://doi.org/10.1145/3033274.3085109}
  {\path{doi:10.1145/3033274.3085109}}.

\bibitem[CEP{\etalchar{+}}23]{cummings2023optimal}
Rachel Cummings, Hadi Elzayn, Emmanouil Pountourakis, Vasilis Gkatzelis, and
  Juba Ziani.
\newblock Optimal data acquisition with privacy-aware agents.
\newblock In {\em IEEE Conference on Secure and Trustworthy Machine Learning
  (SaTML)}, pages 210--224, 2023.
\newblock \href {https://doi.org/10.1109/SaTML54575.2023.00023}
  {\path{doi:10.1109/SaTML54575.2023.00023}}.

\bibitem[CGMS26]{ChaudhuryGMS26}
Bhaskar~Ray Chaudhury, Jugal Garg, Aniket Murhekar, and Jiaxin Song.
\newblock Data pricing via competitive equilibrium.
\newblock In {\em Proceedings of the {ACM} on Web Conference (WWW)}, 2026.

\bibitem[CIL{\etalchar{+}}18]{chen2018optimal}
Yiling Chen, Nicole Immorlica, Brendan Lucier, Vasilis Syrgkanis, and Juba
  Ziani.
\newblock Optimal data acquisition for statistical estimation.
\newblock In {\em ACM Conf.\ Economics and Computation (EC)}, pages 27--44,
  2018.

\bibitem[CK04]{chekuri2004maximum}
Chandra Chekuri and Amit Kumar.
\newblock Maximum coverage problem with group budget constraints and
  applications.
\newblock In {\em International Workshop on Randomization and Approximation
  Techniques in Computer Science}, pages 72--83. Springer, 2004.
\newblock \href {https://doi.org/doi.org/10.1007/978-3-540-27821-4_7}
  {\path{doi:doi.org/10.1007/978-3-540-27821-4_7}}.

\bibitem[CKK24]{ChenKK24}
Xi~Chen, Christian Kroer, and Rachitesh Kumar.
\newblock The complexity of pacing for second-price auctions.
\newblock {\em Mathematics of Operations Research}, 49(4):2109--2135, 2024.
\newblock \href {https://doi.org/10.1287/moor.2022.0009}
  {\path{doi:10.1287/moor.2022.0009}}.

\bibitem[CKP{\etalchar{+}}22]{ConitzerKPSSSW22}
Vincent Conitzer, Christian Kroer, Debmalya Panigrahi, Okke Schrijvers, Eric
  Sodomka, Nicolas~E. Stier-Moses, and Chris Wilkens.
\newblock Pacing equilibrium in first-price auction markets.
\newblock {\em Management Science}, 68(12):8515--8535, 2022.
\newblock \href {https://doi.org/10.1287/mnsc.2022.4310}
  {\path{doi:10.1287/mnsc.2022.4310}}.

\bibitem[CLR{\etalchar{+}}15]{cummings2015accuracy}
Rachel Cummings, Katrina Ligett, Aaron Roth, Zhiwei~Steven Wu, and Juba Ziani.
\newblock Accuracy for sale: Aggregating data with a variance constraint.
\newblock In {\em Symp.\ Innovations in Theoret.\ Computer Science (ITCS)},
  pages 317--324, 2015.
\newblock \href {https://doi.org/10.1145/2688073.2688106}
  {\path{doi:10.1145/2688073.2688106}}.

\bibitem[CPV05]{codenotti2005polynomial}
Bruno Codenotti, Sriram Pemmaraju, and Kasturi Varadarajan.
\newblock On the polynomial time computation of equilibria for certain exchange
  economies.
\newblock In {\em Symp.\ Discrete Algorithms (SODA)}, pages 72--81, 2005.

\bibitem[CPY17]{chen2017complexity}
Xi~Chen, Dimitris Paparas, and Mihalis Yannakakis.
\newblock The complexity of non-monotone markets.
\newblock {\em Journal of the ACM (JACM)}, 64(3):1--56, 2017.
\newblock \href {https://doi.org/10.1145/3064810} {\path{doi:10.1145/3064810}}.

\bibitem[CS06]{conitzer2006computing}
Vincent Conitzer and Tuomas Sandholm.
\newblock Computing the optimal strategy to commit to.
\newblock In {\em ACM Conf.\ Electronic Commerce (EC)}, pages 82--90, 2006.
\newblock \href {https://doi.org/10.1145/1134707.1134717}
  {\path{doi:10.1145/1134707.1134717}}.

\bibitem[CSVY06]{codenotti2006leontief}
Bruno Codenotti, Amin Saberi, Kasturi Varadarajan, and Yinyu Ye.
\newblock Leontief economies encode nonzero sum two-player games.
\newblock In {\em Symp.\ Discrete Algorithms (SODA)}, volume~6, pages 659--667,
  2006.
\newblock URL: \url{https://dl.acm.org/doi/10.5555/1109557.1109629}.

\bibitem[CT05]{cover2005entropy}
Thomas~M. Cover and Joy~A. Thomas.
\newblock Entropy, relative entropy, and mutual information.
\newblock In {\em Elements of Information Theory}, pages 13--55. Wiley, 2005.
\newblock \href {https://doi.org/10.1002/047174882X.ch2}
  {\path{doi:10.1002/047174882X.ch2}}.

\bibitem[CT09]{ChenT09}
Xi~Chen and Shang{-}Hua Teng.
\newblock Spending is not easier than trading: {O}n the computational
  equivalence of {F}isher and {A}rrow-{D}ebreu equilibria.
\newblock In {\em Intl.\ Symp.\ Algorithms and Computation (ISAAC)}, pages
  647--656, 2009.
\newblock \href {https://doi.org/10.1007/978-3-642-10631-6_66}
  {\path{doi:10.1007/978-3-642-10631-6_66}}.

\bibitem[CV21]{cai2020sell}
Yang Cai and Grigoris Velegkas.
\newblock How to sell information optimally: An algorithmic study.
\newblock In {\em Symp.\ Innovations in Theoret.\ Computer Science (ITCS)},
  2021.
\newblock \href {https://doi.org/10.4230/LIPIcs.ITCS.2021.81}
  {\path{doi:10.4230/LIPIcs.ITCS.2021.81}}.

\bibitem[DFHM24]{DeligkasFHM24}
Argyrios Deligkas, John Fearnley, Alexandros Hollender, and Themistoklis
  Melissourgos.
\newblock Constant inapproximability for {F}isher markets.
\newblock In {\em ACM Conf.\ Economics and Computation (EC)}, 2024.
\newblock \href {https://doi.org/10.1145/3670865.3673533}
  {\path{doi:10.1145/3670865.3673533}}.

\bibitem[Din07]{dinur2007pcp}
Irit Dinur.
\newblock The {PCP} theorem by gap amplification.
\newblock {\em Journal of the ACM}, 54:12--es, 2007.
\newblock \href {https://doi.org/10.1145/1236457.1236459}
  {\path{doi:10.1145/1236457.1236459}}.

\bibitem[DPSV08]{DevanurPSV08}
Nikhil Devanur, Christos Papadimitriou, Amin Saberi, and Vijay Vazirani.
\newblock Market equilibrium via a primal--dual algorithm for a convex program.
\newblock {\em J. ACM}, 55(5), 2008.
\newblock \href {https://doi.org/10.1145/1411509.1411512}
  {\path{doi:10.1145/1411509.1411512}}.

\bibitem[DV12]{dobzinski2012query}
Shahar Dobzinski and Jan Vondr{\'a}k.
\newblock From query complexity to computational complexity.
\newblock In {\em Symp.\ Theory of Computing (STOC)}, pages 1107--1116, 2012.
\newblock \href {https://doi.org/10.1145/2213977.2214076}
  {\path{doi:10.1145/2213977.2214076}}.

\bibitem[EG59]{eisenberg1959consensus}
Edmund Eisenberg and David Gale.
\newblock Consensus of subjective probabilities: The pari-mutuel method.
\newblock {\em The Annals of Mathematical Statistics}, 30(1):165--168, 1959.
\newblock URL: \url{http://www.jstor.org/stable/2237130}.

\bibitem[Eis61]{Eisenberg61}
Edmund Eisenberg.
\newblock Aggregation of utility functions.
\newblock {\em Management Sci.}, 7(4):337--350, 1961.
\newblock \href {https://doi.org/10.1287/mnsc.7.4.337}
  {\path{doi:10.1287/mnsc.7.4.337}}.

\bibitem[Fei98]{10.1145/285055.285059}
Uriel Feige.
\newblock A threshold of {$\ln n$} for approximating set cover.
\newblock {\em J. ACM}, 45(4):634–652, 1998.
\newblock \href {https://doi.org/10.1145/285055.285059}
  {\path{doi:10.1145/285055.285059}}.

\bibitem[FGL23]{finster2023substitutes}
Simon Finster, Paul Goldberg, and Edwin Lock.
\newblock Substitutes markets with budget constraints: solving for competitive
  and optimal prices.
\newblock In {\em Conf.\ Web and Internet Economics (WINE)}, 2023.
\newblock \href {https://arxiv.org/abs/2310.03692} {\path{arXiv:2310.03692}}.

\bibitem[FGMS06]{fleischer2006tight}
Lisa Fleischer, Michel~X Goemans, Vahab~S Mirrokni, and Maxim Sviridenko.
\newblock Tight approximation algorithms for maximum general assignment
  problems.
\newblock In {\em Symp.\ Discrete Algorithms (SODA)}, volume~6, pages 611--620,
  2006.

\bibitem[FMMO22]{fallah2022bridging}
Alireza Fallah, Ali Makhdoumi, Azarakhsh Malekian, and Asuman Ozdaglar.
\newblock Bridging central and local differential privacy in data acquisition
  mechanisms.
\newblock {\em Conf.\ Adv.\ Neural Information Processing Systems (NeurIPS)},
  35:21628--21639, 2022.

\bibitem[FMMO24]{fallah2024optimal}
Alireza Fallah, Ali Makhdoumi, Azarakhsh Malekian, and Asuman Ozdaglar.
\newblock Optimal and differentially private data acquisition: Central and
  local mechanisms.
\newblock {\em Operations Research}, 72(3):1105--1123, 2024.
\newblock \href {https://doi.org/10.1287/opre.2022.0014}
  {\path{doi:10.1287/opre.2022.0014}}.

\bibitem[FNS11]{feldman2011unified}
Moran Feldman, Joseph Naor, and Roy Schwartz.
\newblock A unified continuous greedy algorithm for submodular maximization.
\newblock In {\em Symp.\ Foundations of Computer Science (FOCS)}, pages
  570--579. IEEE, 2011.
\newblock \href {https://doi.org/10.1109/FOCS.2011.46}
  {\path{doi:10.1109/FOCS.2011.46}}.

\bibitem[FOT23]{fleckenstein2023review}
Mike Fleckenstein, Ali Obaidi, and Nektaria Tryfona.
\newblock A review of data valuation approaches and building and scoring a data
  valuation model.
\newblock {\em Harvard Data Science Review}, 5(1), 2023.
\newblock \href {https://doi.org/10.1162/99608f92.c18db966}
  {\path{doi:10.1162/99608f92.c18db966}}.

\bibitem[FSVV25]{farboodi2025valuing}
Maryam Farboodi, Dhruv Singal, Laura Veldkamp, and Venky Venkateswaran.
\newblock Valuing financial data.
\newblock {\em The Review of Financial Studies}, 38(3):938--980, 2025.
\newblock \href {https://doi.org/10.1093/rfs/hhae034}
  {\path{doi:10.1093/rfs/hhae034}}.

\bibitem[FV06]{feige2006approximation}
Uriel Feige and Jan Vondr{\'a}k.
\newblock Approximation algorithms for allocation problems: Improving the
  factor of {$1-1/e$}.
\newblock In {\em Symp.\ Foundations of Computer Science (FOCS)}, pages
  667--676. IEEE, 2006.
\newblock \href {https://doi.org/10.1109/FOCS.2006.14}
  {\path{doi:10.1109/FOCS.2006.14}}.

\bibitem[FV23]{farboodi2023data}
Maryam Farboodi and Laura Veldkamp.
\newblock Data and markets.
\newblock {\em Annual Review of Economics}, 15(1):23--40, 2023.
\newblock \href {https://doi.org/10.1146/annurev-economics-082322-023244}
  {\path{doi:10.1146/annurev-economics-082322-023244}}.

\bibitem[Gal60]{Gale60}
D.~Gale.
\newblock {\em Theory of Linear Economic Models}.
\newblock McGraw Hill, N.Y., 1960.

\bibitem[GHL{\etalchar{+}}23]{garg2023approximating}
Jugal Garg, Edin Husi{\'c}, Wenzheng Li, L{\'a}szl{\'o}~A V{\'e}gh, and Jan
  Vondr{\'a}k.
\newblock Approximating {N}ash social welfare by matching and local search.
\newblock In {\em Symp.\ Theory of Computing (STOC)}, pages 1298--1310, 2023.
\newblock \href {https://doi.org/10.1145/3564246.3585255}
  {\path{doi:10.1145/3564246.3585255}}.

\bibitem[GMVY17]{GargMVY17}
Jugal Garg, Ruta Mehta, Vijay~V. Vazirani, and Sadra Yazdanbod.
\newblock Settling the complexity of {L}eontief and {PLC} exchange markets
  under exact and approximate equilibria.
\newblock In {\em Symp.\ Theory of Computing (STOC)}, pages 890--901, 2017.
\newblock \href {https://doi.org/10.1145/3055399.3055474}
  {\path{doi:10.1145/3055399.3055474}}.

\bibitem[GR11]{ghosh2011selling}
Arpita Ghosh and Aaron Roth.
\newblock Selling privacy at auction.
\newblock In {\em ACM Conf.\ Electronic Commerce (EC)}, pages 199--208, 2011.
\newblock \href {https://doi.org/10.1145/1993574.1993605}
  {\path{doi:10.1145/1993574.1993605}}.

\bibitem[GZG22]{GoktasZG22}
Denizalp Goktas, Sadie Zhao, and Amy Greenwald.
\newblock Zero-sum stochastic stackelberg games.
\newblock In {\em Conf.\ Adv.\ Neural Information Processing Systems
  (NeurIPS)}, 2022.

\bibitem[HC24]{Hossain024}
Safwan Hossain and Yiling Chen.
\newblock Equilibrium of data markets with externality.
\newblock In {\em {ICML}}. {PMLR}, 2024.
\newblock URL: \url{https://proceedings.mlr.press/v235/hossain24a.html}.

\bibitem[HK12]{huber2012towards}
Anna Huber and Vladimir Kolmogorov.
\newblock Towards minimizing {$k$}-submodular functions.
\newblock In {\em International symposium on combinatorial optimization}, pages
  451--462. Springer, 2012.
\newblock \href {https://doi.org/10.1007/978-3-642-32147-4_40}
  {\path{doi:10.1007/978-3-642-32147-4_40}}.

\bibitem[HS16]{horner2016selling}
Johannes H{\"o}rner and Andrzej Skrzypacz.
\newblock Selling information.
\newblock {\em Journal of Political Economy}, 124(6):1515--1562, 2016.

\bibitem[Ich21]{ichihashi2021competing}
Shota Ichihashi.
\newblock Competing data intermediaries.
\newblock {\em The RAND Journal of Economics}, 52(3):515--537, 2021.
\newblock \href {https://doi.org/10.1111/1756-2171.12382}
  {\path{doi:10.1111/1756-2171.12382}}.

\bibitem[ITY16]{10.5555/2884435.2884465}
Satoru Iwata, Shin-ichi Tanigawa, and Yuichi Yoshida.
\newblock Improved approximation algorithms for {$k$}-submodular function
  maximization.
\newblock In {\em Symp.\ Discrete Algorithms (SODA)}, page 404–413, 2016.
\newblock \href {https://doi.org/10.1137/1.9781611974331.ch30}
  {\path{doi:10.1137/1.9781611974331.ch30}}.

\bibitem[KCP10]{KorzhykCP10}
Dmytro Korzhyk, Vincent Conitzer, and Ronald Parr.
\newblock Complexity of computing optimal stackelberg strategies in security
  resource allocation games.
\newblock In {\em Conf.\ Artif.\ Intell.\ (AAAI)}, pages 805--810, 2010.
\newblock \href {https://doi.org/10.1609/aaai.v24i1.7638}
  {\path{doi:10.1609/aaai.v24i1.7638}}.

\bibitem[Kle04]{klemperer2004auctions}
Paul Klemperer.
\newblock {\em Auctions: Theory and Practice}.
\newblock Princeton University Press, 2004.

\bibitem[KLMM08]{khot2005inapproximability}
Subhash Khot, Richard~J. Lipton, Evangelos Markakis, and Aranyak Mehta.
\newblock Inapproximability results for combinatorial auctions with submodular
  utility functions.
\newblock {\em Algorithmica}, 52:3--18, 2008.
\newblock \href {https://doi.org/10.1007/s00453-007-9105-7}
  {\path{doi:10.1007/s00453-007-9105-7}}.

\bibitem[MDJM21]{mehta2021sell}
Sameer Mehta, Milind Dawande, Ganesh Janakiraman, and Vijay Mookerjee.
\newblock How to sell a data set? pricing policies for data monetization.
\newblock {\em Information Systems Research}, 32(4):1281--1297, 2021.
\newblock \href {https://doi.org/10.1287/isre.2021.1027}
  {\path{doi:10.1287/isre.2021.1027}}.

\bibitem[MS13]{MehtaS13}
Ruta Mehta and Milind~A. Sohoni.
\newblock Exchange markets: Strategy meets supply-awareness.
\newblock In {\em Conf.\ Web and Internet Economics (WINE)}, volume 8289, pages
  361--362, 2013.
\newblock \href {https://doi.org/10.1007/978-3-642-45046-4_29}
  {\path{doi:10.1007/978-3-642-45046-4_29}}.

\bibitem[MSV08]{mirrokni2008tight}
Vahab Mirrokni, Michael Schapira, and Jan Vondr{\'a}k.
\newblock Tight information-theoretic lower bounds for welfare maximization in
  combinatorial auctions.
\newblock In {\em ACM Conf.\ Electronic Commerce (EC)}, pages 70--77, 2008.
\newblock \href {https://doi.org/10.1145/1386790.1386805}
  {\path{doi:10.1145/1386790.1386805}}.

\bibitem[MYC{\etalchar{+}}23]{MurhekarYCLM23}
Aniket Murhekar, Zhuowen Yuan, Bhaskar~Ray Chaudhury, Bo~Li, and Ruta Mehta.
\newblock Incentives in federated learning: Equilibria, dynamics, and
  mechanisms for welfare maximization.
\newblock In {\em Conf.\ Adv.\ Neural Information Processing Systems
  (NeurIPS)}, 2023.

\bibitem[Mye81]{myerson1981optimal}
Roger~B Myerson.
\newblock Optimal auction design.
\newblock {\em Mathematics of operations research}, 6(1):58--73, 1981.
\newblock \href {https://doi.org/10.1287/moor.6.1.58}
  {\path{doi:10.1287/moor.6.1.58}}.

\bibitem[Nul26]{brightdata2026datamarketplaces}
Jake Nulty.
\newblock Top 15 data marketplaces of 2026: Best platforms ranked, 2026.
\newblock Accessed: 2026-01-28.
\newblock URL:
  \url{https://brightdata.com/blog/web-data/best-data-marketplaces}.

\bibitem[NVX14]{nissim2014redrawing}
Kobbi Nissim, Salil Vadhan, and David Xiao.
\newblock Redrawing the boundaries on purchasing data from privacy-sensitive
  individuals.
\newblock In {\em Symp.\ Innovations in Theoret.\ Computer Science (ITCS)},
  pages 411--422, 2014.

\bibitem[NWF78]{nemhauser1978analysis}
George~L Nemhauser, Laurence~A Wolsey, and Marshall~L Fisher.
\newblock An analysis of approximations for maximizing submodular set
  functions—i.
\newblock {\em Mathematical programming}, 14(1):265--294, 1978.
\newblock \href {https://doi.org/10.1007/BF01588971}
  {\path{doi:10.1007/BF01588971}}.

\bibitem[Orl10]{Orlin10}
James Orlin.
\newblock Improved algorithms for computing {F}isher's market clearing prices.
\newblock In {\em Symp.\ Theory of Computing (STOC)}, pages 291--300, 2010.
\newblock \href {https://doi.org/10.1145/1806689.1806731}
  {\path{doi:10.1145/1806689.1806731}}.

\bibitem[Pei20]{pei2020survey}
Jian Pei.
\newblock A survey on data pricing: from economics to data science.
\newblock {\em IEEE Transactions on knowledge and Data Engineering},
  34(10):4586--4608, 2020.
\newblock \href {https://doi.org/10.1109/TKDE.2020.3045927}
  {\path{doi:10.1109/TKDE.2020.3045927}}.

\bibitem[Rou01]{Roughgarden01}
Tim Roughgarden.
\newblock Stackelberg scheduling strategies.
\newblock In {\em Symp.\ Theory of Computing (STOC)}, pages 104--113. {ACM},
  2001.
\newblock \href {https://doi.org/10.1145/380752.380783}
  {\path{doi:10.1145/380752.380783}}.

\bibitem[Rub18]{Rubinstein18}
Aviad Rubinstein.
\newblock Inapproximability of {N}ash equilibrium.
\newblock {\em {SIAM} J. Comput.}, 47(3):917--959, 2018.
\newblock \href {https://doi.org/10.1137/15M1039274}
  {\path{doi:10.1137/15M1039274}}.

\bibitem[Shm09]{Shmyrev09}
Vadim Shmyrev.
\newblock An algorithm for finding equilibrium in the linear exchange model
  with fixed budgets.
\newblock {\em J. Appl.\ Indust.\ Math.}, 3(4):505--518, 2009.
\newblock \href {https://doi.org/10.1134/S1990478909040097}
  {\path{doi:10.1134/S1990478909040097}}.

\bibitem[SKSC25]{song2025existence}
Jiaxin Song, Pooja Kulkarni, Parnian Shahkar, and Bhaskar~Ray Chaudhury.
\newblock On the existence and complexity of core-stable data exchanges.
\newblock {\em arXiv preprint arXiv:2509.16450}, 2025.

\bibitem[ST20]{schoenebeck2020influence}
Grant Schoenebeck and Biaoshuai Tao.
\newblock Influence maximization on undirected graphs: Toward closing the
  {$(1-1/e)$} gap.
\newblock {\em ACM Transactions on Economics and Computation}, 8:22:1--22:36,
  2020.
\newblock \href {https://doi.org/10.1145/3417748} {\path{doi:10.1145/3417748}}.

\bibitem[{Sug}26]{suger_snowflake_pricing}
{Suger.io Documentation}.
\newblock Snowflake marketplace product pricing plans, 2026.
\newblock Accessed: 2026-01-28.
\newblock URL: \url{https://doc.suger.io/snowflake-marketplace/pricing_plans/}.

\bibitem[Var09]{varian2009economic}
Hal~R Varian.
\newblock Economic aspects of personal privacy.
\newblock In {\em Internet Policy and Economics: Challenges and Perspectives},
  pages 101--109. Springer, 2009.

\bibitem[V{\'e}g16]{Vegh16}
L{\'a}szl{\'o}~A. V{\'e}gh.
\newblock A strongly polynomial algorithm for a class of minimum-cost flow
  problems with separable convex objectives.
\newblock {\em {SIAM} J. Comput.}, 45(5):1729--1761, 2016.
\newblock \href {https://doi.org/10.1137/140978296}
  {\path{doi:10.1137/140978296}}.

\bibitem[Vel23]{veldkamp2023valuing}
Laura Veldkamp.
\newblock Valuing data as an asset.
\newblock {\em Review of Finance}, 27(5):1545--1562, 2023.
\newblock \href {https://doi.org/10.1093/rof/rfac073}
  {\path{doi:10.1093/rof/rfac073}}.

\bibitem[vN28]{neumann1928}
J.~von Neumann.
\newblock Zur theorie der gesellschaftsspiele.
\newblock {\em Mathematische Annalen}, 100(1):295--320, 1928.
\newblock \href {https://doi.org/10.1007/BF01448847}
  {\path{doi:10.1007/BF01448847}}.

\bibitem[Von08]{vondrak2008optimal}
Jan Vondr{\'a}k.
\newblock Optimal approximation for the submodular welfare problem in the value
  oracle model.
\newblock In {\em Symp.\ Theory of Computing (STOC)}, pages 67--74, 2008.
\newblock \href {https://doi.org/10.1145/1374376.1374389}
  {\path{doi:10.1145/1374376.1374389}}.

\bibitem[vS34]{von1934marktform}
Heinrich von Stackelberg.
\newblock {\em Marktform und Gleichgewicht}.
\newblock Springer, Vienna, 1934.
\newblock English translation: \emph{The Theory of the Market Economy}, Oxford
  University Press, 1952.
\newblock \href {https://doi.org/10.2307/2224643} {\path{doi:10.2307/2224643}}.

\bibitem[VSZ10]{vonstengel2010leadership}
B.~Von~Stengel and S.~Zamir.
\newblock Leadership games with convex strategy sets.
\newblock {\em Games and Economic Behavior}, 69(2):446--457, 2010.
\newblock \href {https://doi.org/10.1016/j.geb.2009.11.008}
  {\path{doi:10.1016/j.geb.2009.11.008}}.

\bibitem[WZ16]{ward2016maximizing}
Justin Ward and Stanislav Zivny.
\newblock Maximizing {$k$}-submodular functions and beyond.
\newblock {\em ACM Transactions on Algorithms (TALG)}, 12(4):1--26, 2016.
\newblock \href {https://doi.org/10.1145/2850419} {\path{doi:10.1145/2850419}}.

\end{thebibliography}
\end{document}